%% file: revised_2oct.tex
\documentclass[aps,pra,reprint,superscriptaddress]{revtex4-2}

\usepackage{amsmath,amssymb,graphicx}
\usepackage{url}
\usepackage{amsthm}
\usepackage{orcidlink}
\usepackage{braket}
\usepackage{hyperref}
\hypersetup{breaklinks=true, pdfborder={0 0 0}, colorlinks=true, citecolor=magenta, linkcolor=blue,
    urlcolor=violet}
\usepackage[normalem]{ulem}
\newtheorem{theorem}{Theorem}[section]
\newtheorem{lemma}[theorem]{Lemma}
\newtheorem{corollary}[theorem]{Corollary}
\usepackage{tikz}
\usetikzlibrary{arrows.meta,positioning}

\begin{document}

\title{Complementarity Beyond Definite Causal Order}

\author{Mohd Asad Siddiqui\;\orcidlink{0000-0001-5003-7571}}
\email{asad@ctp-jamia.res.in}
\affiliation{Jawaharlal Nehru Rajkeeya Mahavidyalaya, Sri Vijaya Puram, 744104, Andaman and Nicobar Islands, India}

\author{Md Qutubuddin\;\orcidlink{0000-0002-1294-5154}}
\email{qutubuddinjmi@gmail.com}
\affiliation{Beijing Computational Science Research Center, Beijing 100193, China}

\author{Tabish Qureshi\;\orcidlink{0000-0002-8452-1078}}
\email{tabish@ctp-jamia.res.in}
\affiliation{Centre for Theoretical Physics, Jamia Millia Islamia, New Delhi 110025, India}

\begin{abstract}
Wave--particle duality is a cornerstone of quantum mechanics, traditionally formulated under definite causal order. We investigate how complementarity is modified when the temporal order of operations is coherently superposed, as in the quantum switch. We show that no universal, state-independent, linear additive complementarity relation exists that combines spatial coherence, path distinguishability, and causal coherence within a common unit-normalized bound. This result follows from explicit extremal constructions in the quantum-switch family: although the three quantities are derived from the same global switch state, the standard spatial duality relation can be saturated while the reduced control coherence is simultaneously maximal. Tracing out the order qubit recovers the standard duality relation at the level of the reduced quanton--detector state, while coherence associated with alternative causal orders remains inaccessible in this reduced description. To characterize this contribution, we introduce \emph{causal coherence}, defined as the coherence of the reduced order qubit in the causal-order basis. We further show that complementarity admits a state-dependent entropic formulation based on incompatible measurements on the causal degree of freedom, which arises from a universal uncertainty principle and provides a canonical operationally meaningful description. These results show that operational formulations of complementarity can depend on causal structure and cannot, in general, be fully characterized at the level of reduced quantum states alone.
\end{abstract}

\maketitle

\section{Introduction}

Wave--particle duality lies at the heart of quantum mechanics and provides one of the earliest manifestations of complementarity \cite{starke2026}. In interferometric experiments, complementarity is expressed through trade-off relations between interference visibility and which-path information \cite{WoottersZurek,GreenbergerYasin,Englert,Duerr2001,SiddiquiQureshi2015}. More recent developments have reformulated this relation using information-theoretic tools, where the wave nature is quantified by quantum coherence and the particle nature by path distinguishability \cite{Bera2015,Bagan2016,QureshiSiddiqui2017,Qureshi2019,SiddiquiQureshi2021}. These approaches build upon the resource theory of quantum coherence introduced in \cite{Baumgratz2014,Streltsov2017}, thereby providing an operational framework that extends naturally to mixed states and multi-path interferometers.

Extensions incorporating entanglement, quantum memory, and additional degrees of freedom (such as internal states and environments) further show that correlations can modify and, in some cases, tighten duality relations \cite{JakobBergou2007,Bu2018,Roy2022,Sun2025,Banaszek2013,Sharma2020}.

A central assumption underlying these formulations is that the relevant operations occur in a \emph{definite causal order}. In standard interferometric scenarios, the which-path interaction and the interference measurement are applied sequentially within a fixed temporal structure. This raises the question of how wave--particle complementarity is reshaped when the temporal order of operations is coherently superposed.

Quantum theory allows processes in which the temporal order of operations is not fixed but can exist in a coherent superposition of causal orders \cite{Oreshkov2012,Chiribella2013,Brukner2014}. Such processes are described within the process-matrix framework, which generalizes quantum theory beyond fixed causal structures \cite{Oreshkov2012}. A paradigmatic physically realizable example is the quantum switch, in which two operations are applied in a superposition of causal orders \cite{Chiribella2013}.

Experimental realizations have demonstrated coherent control of causal order \cite{Procopio2015,Rubino2017,Goswami2018,Goswami2020,Deng2025}, while theoretical studies have identified advantages in communication and channel discrimination tasks \cite{Chiribella2012,Ebler2018,Chiribella2021}.

Despite these developments, the implications of indefinite causal order (ICO) for foundational principles---particularly wave--particle complementarity---remain largely unexplored at a quantitative and operational level.

One may therefore anticipate a generalized complementarity relation involving coherence, path distinguishability, and interference between causal orders. We show, however, that no universal linear additive complementarity relation exists within the quantum switch framework, and therefore complementarity does not admit a universal linear additive description in this setting.

To investigate this, we formulate wave--particle duality in a quantum switch, where the temporal order between the which-path interaction and the interference operation is coherently controlled. Tracing out the order qubit recovers the standard duality relation at the level of the reduced quanton--detector state, but removes access to the control degree of freedom and to any order--quanton--detector correlations not contained in that reduced state.

To characterize this missing contribution, we introduce \emph{causal coherence}, defined as the coherence of the reduced order qubit in the causal-order basis.

We construct explicit processes in which the standard spatial duality relation is saturated while causal coherence is simultaneously maximal, thereby ruling out any universal joint tradeoff of the form considered here.

This failure motivates an entropic formulation of complementarity. Although wave--particle duality is known to be equivalent to an entropic uncertainty relation in standard interferometric settings \cite{Coles2014,Coles2017}, we show that this correspondence does not extend to scenarios with indefinite causal order. Instead, complementarity admits a state-dependent entropic uncertainty formulation based on incompatible measurements on the causal degree of freedom, thereby incorporating causal structure beyond fixed-order scenarios.

These results show that operational formulations of complementarity depend on the causal structure of the underlying quantum process and, in general, cannot be fully characterized from reduced quantum states alone.

The paper is organized as follows. In Sec.~\ref{sec:fixed} we review wave--particle duality for fixed causal order. In Sec.~\ref{sec:indef} we formulate the quantum switch scenario. In Sec.~\ref{sec:causal} we introduce causal coherence and establish its operational meaning. In Sec.~\ref{sec:no_go} we prove the no-go theorem. In Sec.~\ref{sec:entr} we derive the entropic formulation. Finally, we conclude in Sec.~\ref{sec:conc}.

\section{Wave--Particle Duality for Fixed Causal Order}\label{sec:fixed}

We review the resource-theoretic formulation of wave--particle duality for interferometers with a definite temporal order. Consider an $n$-path interferometer in which a quanton is prepared in the pure state
\begin{equation}
\ket{\Psi}=\sum_{i=1}^n \sqrt{p_i}\,e^{i\phi_i}\ket{\psi_i},
\qquad
\sum_i p_i=1 .
\end{equation}
where $p_i$ are the probabilities associated with each interferometric path, and $\{\ket{\psi_i}\}$ forms an orthonormal path basis.

Which-path information is extracted by coupling the quanton to a detector initially prepared in the reference state $\ket{d_0}$. The interaction correlates each path with a detector state,
\begin{equation}
\ket{\psi_i}\ket{d_0}\longrightarrow \ket{\psi_i}\ket{d_i}.
\end{equation}
After the interaction, the joint quanton--detector state reads
\begin{equation}
\ket{\Psi}_{QD}=\sum_i \sqrt{p_i}\,e^{i\phi_i}\ket{\psi_i}\ket{d_i}.
\end{equation}

The corresponding density operator is
\begin{equation}
\rho_{QD}=\sum_{i,j} \sqrt{p_i p_j}\,e^{i(\phi_i-\phi_j)}
\ket{\psi_i}\bra{\psi_j}\otimes
\ket{d_i}\bra{d_j}.
\label{eq:fixed_dop}
\end{equation}
Tracing out the detector subsystem yields the reduced state of the quanton
\begin{equation}
\rho_Q
= \operatorname{Tr}_D(\rho_{QD})
= \sum_{i,j} \sqrt{p_i p_j}\,e^{i(\phi_i-\phi_j)}
\ket{\psi_i}\bra{\psi_j}
\langle d_j|d_i\rangle.
\end{equation}
The detector overlaps $\langle d_j|d_i\rangle$ dictate the amount of available which-path information.

The wave character of the quanton is quantified by the normalized $l_1$-norm of coherence \cite{Bera2015,PaulQureshi2017},
\begin{equation}
C
= \frac{1}{n-1}
\sum_{i\neq j}
\sqrt{p_i p_j}\,
\big|\langle d_j|d_i\rangle\big|,
\label{eq:coherence}
\end{equation}
which quantifies residual spatial superposition in the path basis.

The particle character is quantified by the path distinguishability $D_Q$, defined operationally as an upper bound on the probability with which the detector states $\{|d_i\rangle\}$, occurring with prior probabilities $\{p_i\}$, can be unambiguously distinguished. For an $n$-path interferometer, this quantity can be expressed as~\cite{SiddiquiQureshi2015}
\begin{equation}
D_Q
= 1 - \frac{1}{n-1}
\sum_{i\neq j}
\sqrt{p_i p_j}\,
\big|\langle d_j|d_i\rangle\big|.
\end{equation}
This quantity satisfies $0 \le D_Q \le 1$, with $D_Q=1$ corresponding to perfect which-path information (orthogonal detector states) and $D_Q=0$ corresponding to completely indistinguishable detector states with equal path priors.

Hence, for any interferometer with a fixed causal order \cite{Bera2015},
\begin{equation}
C + D_Q \leq 1.
\label{CDQ}
\end{equation}
For pure quanton--detector states the relation can be saturated, $C + D_Q = 1$, while for mixed states the inequality holds.

This trade-off reflects a fundamental limitation imposed by a definite temporal structure.

Examples of definite causal order include standard which-path interferometric setups in which the quanton is first entangled with a detector and interference is subsequently measured.

It is instructive to note that even in scenarios where interference is observed prior to attempts at which-path detection, complementarity remains intact. For example, in the experiment of Afshar \textit{et al.}~\cite{afshar2007}, interference was inferred before path information was probed. Subsequent analyses showed that no genuine which-path information was available for the detected photons \cite{Qureshi2012,Qureshi2023}, and that the standard duality relations were therefore fully respected.

\section{Indefinite Causal Order and the Quantum Switch}\label{sec:indef}

Having established wave--particle duality under a definite causal structure, we now extend this framework to scenarios in which the temporal order of operations is not fixed but can exist in a coherent superposition. In such situations, an additional physical degree of freedom---the order qubit---becomes relevant, enabling coherence between alternative causal orders.

Nevertheless, if the order qubit is traced out, the resulting quanton--detector state reduces to a classical mixture of fixed causal orders, and the standard spatial duality relation remains intact at this level of description.

However, this reduced description does not capture the full process, which involves coherent superpositions of alternative causal orders and is therefore causally nonseparable. A proper description of such processes requires a framework that goes beyond fixed causal structures, namely the process-matrix formalism.

Within this framework, a bipartite quantum process involving two laboratories $A$ and $B$ is described by a process matrix $W$ acting on the tensor product of their input and output Hilbert spaces. The probability of outcomes associated with local operations $\mathcal{M}_A$ and $\mathcal{M}_B$ is given by the generalized Born rule
\begin{equation}
P(\mathcal{M}_A,\mathcal{M}_B)
=
\operatorname{Tr}\!\left[W\left(M_A \otimes M_B\right)\right],
\end{equation}
where $M_A$ and $M_B$ are the Choi--Jamiołkowski operators corresponding to the completely positive maps $\mathcal{M}_A$ and $\mathcal{M}_B$ \cite{Oreshkov2012}.

Process matrices compatible with a definite causal order can be written as convex mixtures of processes with fixed order $A\prec B$ and $B\prec A$, whereas those that cannot admit such a decomposition are termed \emph{causally nonseparable}. In particular, the quantum switch provides a physically realizable example of such a causally nonseparable process \cite{Chiribella2013,Procopio2015,Rubino2017,Goswami2018}.

In the setup considered here, a quantum switch coherently controls the temporal order between the which-path interaction and the interference operation, as illustrated schematically in Fig.~\ref{fig:quantum-switch}.

\begin{figure*}[t]
\centering
\includegraphics[width=.9\textwidth]{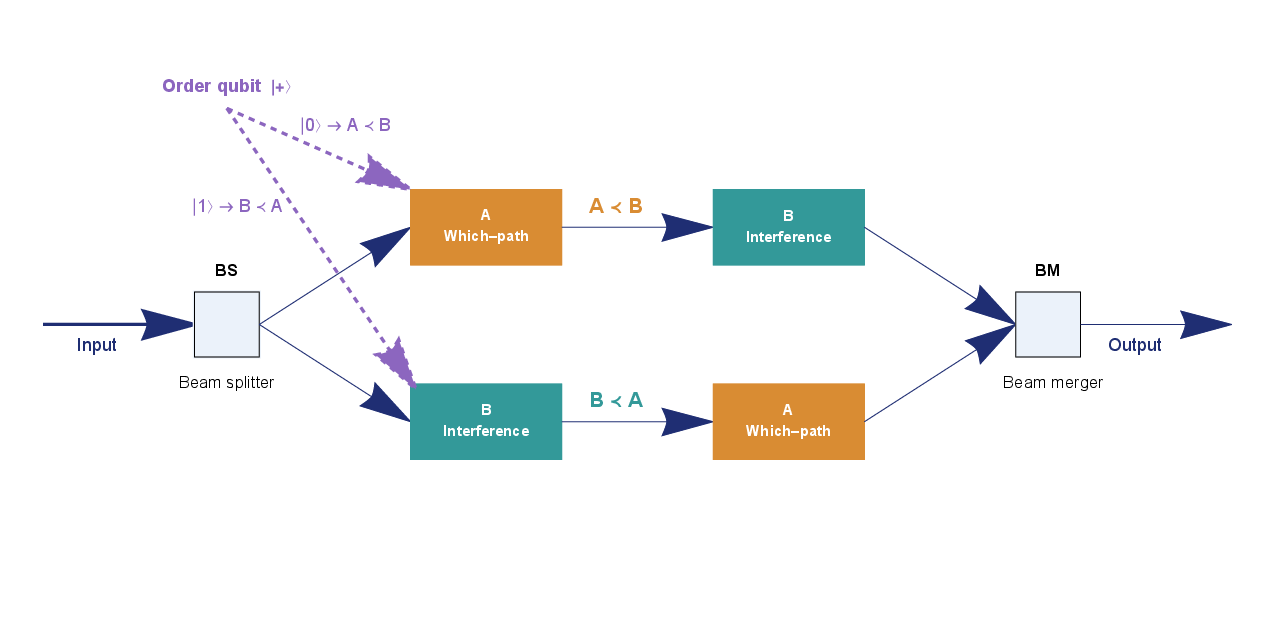}
\vspace{-1em}
\caption{Schematic representation of a quantum switch controlling the temporal order between a which-path interaction $A$ and an interference operation $B$ in an interferometric setup. The order qubit places the two possible orders $A\prec B$ and $B\prec A$ in a superposition of alternative causal orders.}
\label{fig:quantum-switch}
\end{figure*}

Let $\rho^{(0)}_{QD}$ denote the initial joint state of the quanton and the which-path detector. Let $U_A$ represent the which-path interaction and $U_B = U_Q \otimes I_D$ the interference operation acting only on the quanton Hilbert space.

The two definite temporal orders correspond to the unitary evolutions
\begin{align}
\rho_{A\prec B}
&=
U_B U_A \rho^{(0)}_{QD} U_A^\dagger U_B^\dagger, \\
\rho_{B\prec A}
&=
U_A U_B \rho^{(0)}_{QD} U_B^\dagger U_A^\dagger .
\end{align}
These states can differ when $[U_A,U_B]\neq 0$; noncommutativity does not guarantee different output states for every input.

The quantum switch coherently controls the order of these operations via an order qubit $O$. Defining
\begin{align}
U_{A\prec B} &= U_B U_A, \\
U_{B\prec A} &= U_A U_B ,
\end{align}
the controlled unitary implementing the switch is
\begin{equation}
U_{\mathrm{sw}} =
U_{A\prec B} \otimes |0\rangle\langle 0|
+
U_{B\prec A} \otimes |1\rangle\langle 1|.
\end{equation}

For an initial order state $\rho_O$, the global state is
\begin{equation}
\rho_{\mathrm{tot}} =
U_{\mathrm{sw}}
\left(\rho^{(0)}_{QD} \otimes \rho_O\right)
U_{\mathrm{sw}}^\dagger .
\end{equation}

Let the order qubit be prepared in a general state in the basis $\{|0\rangle,|1\rangle\}$,
\begin{equation}
\rho_O = p\,|0\rangle\langle0|
+ (1-p)\,|1\rangle\langle1|
+ \kappa\,|0\rangle\langle1|
+ \kappa^*\,|1\rangle\langle0|,
\end{equation}
with $0 \le p \le 1$ and $|\kappa|^2 \le p(1-p)$.

The state $\rho_{\mathrm{tot}}$ contains both diagonal and off-diagonal contributions in the order-qubit basis, encoding classical mixtures and coherent superpositions of alternative causal orders. Tracing over the order qubit removes the off-diagonal terms, since $\operatorname{Tr}_O(|0\rangle\langle1|)= \operatorname{Tr}_O(|1\rangle\langle0|)=0$, yielding the reduced quanton--detector state
\begin{equation}
\rho_{QD} = \operatorname{Tr}_O(\rho_{\mathrm{tot}})
= p\,\rho_{A\prec B}
+ (1-p)\,\rho_{B\prec A}.
\label{eq:convexrho}
\end{equation}

Operational wave and particle quantities are thus determined from the reduced state $\rho_{QD}$, which obeys the standard wave--particle duality relation. This reduced state contains no off-diagonal order labels. Access to the order qubit is required to determine its reduced coherence, while the full causal-order content may also reside in correlations between the order qubit and the quanton--detector system.

We therefore distinguish coherence at different levels of description. We denote by $C$ the coherence for fixed causal order, while $C_q$ denotes the coherence of the reduced quanton state under indefinite causal order, defined as
\[
\rho_Q = \operatorname{Tr}_{D,O}(\rho_{\mathrm{tot}}).
\]

The corresponding wave coherence is
\begin{equation}
C_q =
\frac{1}{n-1} \sum_{i\neq j}
\left| \langle\psi_i|\rho_Q|\psi_j\rangle \right|.
\end{equation}
This is the normalized $l_1$-norm of coherence of $\rho_Q$ in the path basis $\{|\psi_i\rangle\}$, analogous to Eq.~\eqref{eq:coherence}.

The operational path distinguishability $D_Q^{\mathrm{ICO}}$ is defined for the detector states $\{|d_i\rangle\}$ with prior probabilities $\{p_i\}$, as encoded in the reduced quanton--detector state $\rho_{QD}$.

Convexity of the $l_1$-norm of coherence under mixing of density operators~\cite{Baumgratz2014} implies
\begin{equation}
C_q
\le
p\,C^{A\prec B}
+
(1-p)\,C^{B\prec A}.
\end{equation}

Similarly, the optimal success probability for unambiguous quantum state discrimination is convex under classical mixing of the corresponding sets of detector states (more generally, detector ensembles), yielding
\begin{equation}
D_Q^{\mathrm{ICO}} \le
p D_Q^{A\prec B}
+
(1-p) D_Q^{B\prec A}.
\end{equation}
A proof of this property is provided in Appendix~\ref{app:uqsd_convexity}.

Combining this with the fixed-order wave--particle duality
relations $C^{A\prec B}+D_Q^{A\prec B}\le 1$ and
$C^{B\prec A}+D_Q^{B\prec A}\le 1$~\cite{Bera2015},
we obtain
\begin{equation}
C_q + D_Q^{\mathrm{ICO}} \le 1 .
\end{equation}

Thus, when the order qubit is traced out, the reduced quanton--detector state obeys the standard two-term wave--particle duality relation.

The quantities $C$ and $D_Q$ characterize wave and particle behavior under a fixed causal order, while $C_q$ and $D_Q^{\mathrm{ICO}}$ describe the corresponding quantities for the reduced quanton--detector state in the presence of indefinite causal order.

However, this reduced description does not capture the control degree of freedom or its correlations with the quanton--detector system. In the next section we quantify the coherence of the reduced order qubit in the causal-order basis, while keeping distinct this reduced-state quantity from process-level causal correlations.

\section{Causal Coherence}\label{sec:causal}

We now quantify the coherence of the reduced order qubit in the causal-order basis and establish its operational significance.

We consider a pure quantum switch state of the form
\begin{equation}
|\Psi_{\mathrm{tot}}\rangle =
\sqrt{p}\,|\Psi_{A\prec B}\rangle|0\rangle
+ e^{i\theta}\sqrt{1-p}\,|\Psi_{B\prec A}\rangle|1\rangle,
\label{eq:switch_state}
\end{equation}
where $|\Psi_{A\prec B}\rangle$ and $|\Psi_{B\prec A}\rangle$ are normalized joint quanton--detector states corresponding to the two definite causal orders, and $0 \le p \le 1$.

In the process-matrix framework, the quantum switch realizes a causally nonseparable process in which the control system (the order qubit) coherently selects the causal structure between the operations $A$ and $B$. Tracing out the order qubit yields the classical mixture of fixed causal orders given in Eq.~\eqref{eq:convexrho}. In contrast, tracing over the quanton--detector degrees of freedom yields the reduced order-qubit state
\begin{equation}
\rho_O =
\begin{pmatrix}
p & \kappa \\
\kappa^* & 1-p
\end{pmatrix},
\end{equation}
where the off-diagonal element $\kappa$ is determined by the overlap between the two causal-order states,
\begin{equation}
\kappa = \sqrt{p(1-p)}\,e^{-i\theta}
\langle\Psi_{A\prec B}|\Psi_{B\prec A}\rangle .
\end{equation}

Positivity of $\rho_O$ implies $|\kappa|^2 \le p(1-p)$, which follows from the Cauchy--Schwarz inequality $|\langle \Psi_{A\prec B} | \Psi_{B\prec A} \rangle| \le 1$.

\medskip
\noindent
We define causal coherence as the $l_1$-norm of coherence of the order qubit,
\begin{equation}
C_{\mathrm{causal}}
= 2|\kappa|
= 2\sqrt{p(1-p)}
|\langle \Psi_{A\prec B} | \Psi_{B\prec A} \rangle|.
\label{eq:C_causal}
\end{equation}

This quantity characterizes coherence of the reduced order qubit in the causal-order basis.

From the bound $|\kappa|^2 \le p(1-p)$, it follows that
\[
|\kappa| \le \sqrt{p(1-p)},
\]
and hence
\[
C_{\mathrm{causal}} = 2|\kappa| \le 2\sqrt{p(1-p)} \le 1.
\]
where the last inequality uses $p(1-p)\le \tfrac{1}{4}$. Hence,
\[
0 \le C_{\mathrm{causal}} \le 1.
\]

As a reduced-state quantity, causal coherence vanishes when the switch is prepared in a definite causal order ($p=0$ or $p=1$), or when the two causal-order states are orthogonal. It attains its maximum value $C_{\mathrm{causal}}=1$ when the two causal-order states are identical and the switch is prepared in an equal superposition ($p=\tfrac{1}{2}$).

\medskip

\noindent\textbf{Interpretation and limitations.}

We emphasize that $C_{\mathrm{causal}}$ quantifies the coherence
of the reduced control state in the causal-order basis and is
conceptually distinct from measures of causal nonseparability,
such as those based on the process matrix framework
\cite{Oreshkov2012,Araujo2015,Branciard2016}. Causal
nonseparability quantifies whether, or to what extent, a process
cannot be decomposed as a convex mixture of definite causal
orders, whereas $C_{\mathrm{causal}}$ quantifies coherence of
the reduced order qubit. The two notions characterize different
aspects of the quantum-switch description and need not coincide.

In particular, in the commuting sector $[U_A,U_B]=0$, the two
causal orders generate identical ordered evolutions, and yet
$C_{\mathrm{causal}}$ can be maximal when the order qubit is
prepared in an equal superposition. This observation highlights
that $C_{\mathrm{causal}}$ quantifies local coherence of the
control system, whereas operational distinguishability
characterizes the ability to discriminate the corresponding
ordered output states.

At the process level, however, the control system remains
essential: inclusion of the control is necessary for the quantum
switch to constitute a valid process, and the switch may be
interpreted in terms of entanglement between causal order and
the control~\cite{Costa2022}. We also note that the present work
does not employ causal nonseparability as its resource concept;
instead, $C_{\mathrm{causal}}$ is a reduced control-coherence
quantity that is directly accessible through measurements on
the order qubit.

\medskip

To probe interference between the two causal orders, we measure
the order qubit in the generalized basis
\begin{equation}
|\pm_\phi\rangle = \frac{1}{\sqrt{2}}
\left( |0\rangle \pm e^{i\phi}|1\rangle \right).
\end{equation}

The corresponding outcome probabilities are
\begin{align}
P_\pm(\phi)
&= \langle \pm_\phi | \rho_O | \pm_\phi \rangle
\nonumber\\
&= \frac12 \left( 1 \pm
\bigl( e^{i\phi}\kappa + e^{-i\phi}\kappa^* \bigr) \right).
\label{eq:prob_general}
\end{align}

Writing $\kappa = |\kappa| e^{i\varphi_0}$, with $\varphi_0 = \arg \kappa$, the interference term becomes
\begin{equation}
e^{i\phi}\kappa + e^{-i\phi}\kappa^*
= 2|\kappa|\cos(\phi + \varphi_0).
\end{equation}
Thus,
\begin{equation}
P_\pm(\phi) = \frac12 \left( 1 \pm 2|\kappa|\cos(\phi + \varphi_0) \right).
\label{eq:prob_cosine}
\end{equation}

The phase $\varphi_0=\arg(\kappa)$ determines the phase shift
of the interference pattern, i.e., the locations of its maxima
and minima as a function of $\phi$, while $|\kappa|$ determines
its visibility.

\medskip

The interference visibility of the order qubit is therefore
\begin{equation}
V_{\mathrm{causal}}
=
\frac{P_{\max}-P_{\min}}{P_{\max}+P_{\min}}
=
2|\kappa|.
\end{equation}

Hence, causal coherence admits a direct operational
interpretation as the interference visibility of the reduced
order qubit, establishing it as an experimentally observable
reduced-state quantity.

In this work, causal coherence is defined without conditioning
on specific measurement outcomes. It is the $l_1$-norm of
coherence of the reduced order-qubit state and therefore captures
local control coherence in the causal-order basis. This contrasts
with spatial coherence, which is defined at the level of the
reduced quanton state and characterizes superposition between
interferometric paths. The corresponding post-selected duality
relations are analyzed in a subsequent section.

This perspective admits a natural resource-theoretic
interpretation.

\paragraph*{Resource-theoretic interpretation}
The quantity $C_{\mathrm{causal}}$ can be understood within the standard
coherence resource theory of the reduced order qubit. The incoherent states
for this reduced-state resource are density operators diagonal in the
causal-order basis $\{|0\rangle,|1\rangle\}$,
\[
\rho_O=p|0\rangle\langle0|+(1-p)|1\rangle\langle1|.
\]
In the standard coherence resource theory, incoherent operations with
respect to this basis are taken as the free operations, as defined in
Ref.~\cite{Baumgratz2014}. Under such operations, $C_{\mathrm{causal}}$ is a
monotone of local control coherence. We do not assume that these operations
are the physically free operations for controlling causal order.

When $C_{\mathrm{causal}}=0$, the reduced order-qubit state is diagonal, so no local control coherence is visible in $\rho_O$. This does not imply that the underlying process is causally separable, since correlations between the order qubit and the quanton--detector system may remain. Tracing out the order qubit nevertheless gives
\[
\rho_{QD}=p\,\rho_{A\prec B}+(1-p)\rho_{B\prec A},
\]
as in Eq.~\eqref{eq:convexrho}. Thus, $C_{\mathrm{causal}}$ alone neither determines process-level causal nonseparability nor captures all resources or operational advantages of the full quantum switch.

\paragraph*{Experimental accessibility.}
Causal coherence is experimentally accessible in photonic implementations of the quantum switch, where the order qubit coherently controls the temporal order of operations \cite{Procopio2015,Rubino2017,Goswami2018}. Measuring the order qubit in a superposition basis yields the interference visibility from which $C_{\mathrm{causal}}$ can be extracted. This establishes it as an experimentally measurable property of the reduced control state.

Operational advantages of the quantum switch have also been
established in quantum communication scenarios
\cite{Chiribella2021}. Such advantages concern the full
controlled process and its correlations, and should not be
attributed to the reduced scalar quantity
$C_{\mathrm{causal}}$ alone.

While causal coherence is not accessible at the level of the reduced quanton--detector state, it can be revealed through measurements on the order qubit.

If measurement outcomes are not conditioned upon, i.e., one averages over the post-selected subensembles, the interference terms arising from coherence between causal orders cancel, thereby recovering the reduced state in Eq.~\eqref{eq:convexrho}. Consequently, the effects of interference between causal orders on the quanton--detector subsystem are observable only in conditional (post-selected) statistics.

This mechanism is closely analogous to a quantum eraser \cite{Scully1982,Walborn2002}: the global state may encode which-order information in correlations between the order qubit and the quanton--detector system. Tracing out the order qubit makes these correlations inaccessible in the reduced quanton--detector state, whereas measurement of the order qubit in a superposition basis can erase accessible which-order information and reveal interference in the conditional statistics.

Extensions of the quantum eraser concept to nonlocal (spatial) settings under definite causal order have been investigated in \cite{SiddiquiQureshi2016}. In contrast, the present work concerns interference between alternative causal orders, highlighting a distinct form of coherence associated with temporal order.

We emphasize that $C_q$ denotes the coherence of the reduced quanton state obtained after tracing out the order qubit, whereas $C^{(\pm)}$ denotes coherence conditioned on post-selection. These quantities correspond to different operational levels of description.
\paragraph*{Post-selected duality.}
Measurements of the order qubit in a superposition basis induce a conditional (post-selected) description of the quanton--detector system. The resulting states exhibit modified coherence and distinguishability, reflecting interference between alternative causal orders.

Conditioned on the measurement outcomes $\pm$ of the order qubit, the reduced quanton state takes the form
\begin{equation}
\rho_Q^{(\pm)}
=
\begin{pmatrix}
\frac12 & \gamma_\pm \\
\gamma_\pm^* & \frac12
\end{pmatrix},
\end{equation}
for a symmetric two-path interferometer with equal initial path probabilities, such that post-selection preserves path symmetry (see Appendix~\ref{app:post_selected_derivation} for the symmetry condition and the general case).

The off-diagonal element $\gamma_\pm$ admits a decomposition into contributions from classical mixing of fixed causal orders and interference between alternative causal orders, and is given by
\begin{equation}
\gamma_\pm
=
\frac{1}{4\mathcal{N}_\pm}
\left[
C_{\mathrm{mix}}
\pm
\sqrt{p(1-p)}
\big(
e^{i\theta}\Gamma_{10}
+
e^{-i\theta}\Gamma_{01}^*
\big)
\right].
\end{equation}

Here $C_{\mathrm{mix}} = p\,\gamma_{A\prec B} + (1-p)\,\gamma_{B\prec A}$, with $\gamma_X = \langle d_1^{X} | d_0^{X} \rangle$ for $X \in \{A\prec B,\, B\prec A\}$. The quantities $\Gamma_{ij}=\langle d_i^{A\prec B}|d_j^{B\prec A}\rangle$ denote overlaps between detector states corresponding to different causal orders. The normalization factor $\mathcal{N}_\pm$ is given in Appendix~\ref{app:post_selected_derivation}.

The corresponding wave and particle quantities for the two-path case are
\begin{equation}
C^{(\pm)} = 2|\gamma_\pm|,
\qquad
D^{(\pm)} = 1 - 2|\gamma_\pm|,
\end{equation}
which satisfy the complementarity relation
\begin{equation}
C^{(\pm)} + D^{(\pm)} = 1.
\end{equation}

Thus, the standard complementarity relation is recovered within each post-selected ensemble, while causal coherence manifests through interference contributions in $\gamma_\pm$. Complementarity therefore remains valid within each operational setting. This raises the question of whether it can be extended to a universal linear relation incorporating causal coherence, which we address in the next section.

\medskip
\noindent\textbf{Causal duality.}
In analogy with spatial wave--particle duality, the particle-like character of causal order can be quantified operationally via unambiguous quantum state discrimination (UQSD) between the two causal-order states $|\Psi_{A\prec B}\rangle$ and $|\Psi_{B\prec A}\rangle$, occurring with prior probabilities $p$ and $1-p$.

For two pure states with prior probabilities $p$ and $1-p$ (assuming $p \ge 1/2$ without loss of generality), the optimal success probability is given by the Ivanovic--Dieks--Peres bound~\cite{Ivanovic1987,Dieks1988,Peres1988},
\[
D_{\mathrm{causal}}^{\mathrm{UQSD}} = 1 - 2\sqrt{p(1-p)} \left| \langle \Psi_{A\prec B}|\Psi_{B\prec A}\rangle \right|,
\]
which holds when $\lvert \langle \Psi_{A\prec B}|\Psi_{B\prec A}\rangle \rvert \le \sqrt{(1-p)/p}$~\cite{Jaeger1995}. Outside the Ivanovic--Dieks--Peres regime determined by the overlap and prior probabilities, the optimal discrimination strategy differs and the simple equality $C_{\mathrm{causal}} + D_{\mathrm{causal}}^{\mathrm{UQSD}} = 1$ does not generally hold.

Using $C_{\mathrm{causal}} = 2\sqrt{p(1-p)} \left| \langle \Psi_{A\prec B}|\Psi_{B\prec A}\rangle \right|$, we obtain
\begin{equation}
C_{\mathrm{causal}} + D_{\mathrm{causal}}^{\mathrm{UQSD}} = 1,
\end{equation}
which holds for pure causal-order states within the IDP regime. For mixed states, the relation is generally replaced by an inequality.

This relation is directly analogous to two-path wave--particle duality, but applies to the temporal (causal-order) degree of freedom. It highlights the distinct operational role of causal coherence and motivates the investigation of whether a unified complementarity relation with spatial coherence and path distinguishability can exist.

Although causal-order discrimination, as analyzed via minimum-error discrimination in Appendix~\ref{app:causal_discrimination}, is governed by detector-state overlaps $\langle d_j|d_i\rangle$, this dependence can be expressed in terms of overlap between the global causal-order states $\langle \Psi_{A\prec B} | \Psi_{B\prec A} \rangle$. Accordingly, the UQSD formulation provides a natural process-level characterization of causal distinguishability---free of errors---thereby elevating the description beyond the detector level.

\section{No Universal Additive Complementarity with Causal Coherence}
\label{sec:no_go}

The preceding analysis establishes two distinct forms of complementarity:
(i) spatial wave--particle duality between coherence $C_q$ and path distinguishability $D_Q^{\mathrm{ICO}}$, and 
(ii) causal duality between causal coherence $C_{\mathrm{causal}}$ and the corresponding particle-like quantity $D_{\mathrm{causal}}^{\mathrm{UQSD}}$.

A natural question is whether these two forms can be unified into a single tradeoff relation. A straightforward extension of standard duality relations suggests an additive constraint of the form
\begin{equation}
C_q + D_Q^{\mathrm{ICO}} + C_{\mathrm{causal}} \le 1.
\label{eq:additive_attempt}
\end{equation}
Whether such a relation exists cannot be decided from subsystem separation alone. The reduced quanton--detector state and the reduced control state are marginals of the same global switch state and may therefore be subject to constraints inherited from that state. The relevant question is instead whether the set of physically realizable quantum-switch states obeys the proposed additive bound.

The quantum switch admits explicit realizable counterexamples: there exist configurations in which the spatial duality relation is saturated,
\[
C_q + D_Q^{\mathrm{ICO}} = 1,
\]
while the causal coherence simultaneously attains its maximal value,
\[
C_{\mathrm{causal}} = 1.
\]
The simultaneous attainability of these extrema directly violates Eq.~\eqref{eq:additive_attempt}, thereby excluding any universal, state-independent, linear additive relation of the proposed form.

Standard nonlinear duality relations, such as quadratic relations of the Englert type~\cite{Englert}, arise from specific geometric or information-theoretic constraints.

\medskip

\noindent\textbf{Absence of a universal additive triality structure.}
Unlike known triality relations~\cite{Jakob2010,Qian2018,Roy2022},
the quantities considered here characterize different reduced
degrees of freedom. Nevertheless, the failure of a common
unit-normalized additive bound follows from the realizable
extremal point exhibited above, not from subsystem separation
alone.

\medskip

Although Eq.~\eqref{eq:additive_attempt} is ruled out, one may still
consider weighted additive relations in which the causal-coherence
term is assigned a fixed, state-independent coefficient $\alpha>0$.
The following no-go theorem excludes the entire family of such
relations.

\begin{theorem}[No universal additive complementarity with causal coherence]
There exists no universal state-independent inequality of the form
\begin{equation}
C_q + D_Q^{\mathrm{ICO}} + \alpha\, C_{\mathrm{causal}} \le 1 ,
\label{eq:no_go_linear}
\end{equation}
with any constant $\alpha>0$ that holds for all processes realizable via the quantum switch.
\end{theorem}

\begin{proof}
Assume, for contradiction, that Eq.~\eqref{eq:no_go_linear} holds universally. In particular, it must hold in the commuting sector of the quantum switch.

Consider two unitary operations $U_A$ and $U_B$ acting on $\mathcal{H}_Q \otimes \mathcal{H}_D$, with $U_B = U_Q \otimes I_D$, such that $[U_A,U_B]=0$. For any pure input state $|\Psi^{(0)}\rangle$, define
\[
\ket{\Psi_{A\prec B}} = U_B U_A \ket{\Psi^{(0)}}, 
\qquad
\ket{\Psi_{B\prec A}} = U_A U_B \ket{\Psi^{(0)}}.
\]
Commutativity implies
\[
\ket{\Psi_{A\prec B}} = \ket{\Psi_{B\prec A}} =: \ket{\Phi},
\]
so that their overlap equals unity.

The corresponding quantum switch state is
\[
|\Psi_{\mathrm{tot}}\rangle
=
\sqrt{p}\,|\Phi\rangle|0\rangle
+
e^{i\theta}\sqrt{1-p}\,|\Phi\rangle|1\rangle,
\]
yielding causal coherence
\[
C_{\mathrm{causal}} = 2\sqrt{p(1-p)},
\]
which is maximal for $p=\tfrac12$.

Since the reduced quanton--detector state is $\rho_{QD}=|\Phi\rangle\langle\Phi|$, the spatial quantities reduce to their fixed-order forms. As $U_B U_A$ is unitary, any state $|\Phi\rangle$ can be realized by an appropriate choice of input.

Choose $|\Phi\rangle$ such that the fixed-order duality relation is saturated, $C_q + D_Q^{\mathrm{ICO}} = 1$. Then, for $p=\tfrac12$,
\[
C_q + D_Q^{\mathrm{ICO}} + \alpha C_{\mathrm{causal}} = 1 + \alpha > 1,
\]
contradicting Eq.~\eqref{eq:no_go_linear}. Hence no such universal additive relation exists.
\end{proof}

\medskip

\noindent\textbf{Remark on scope.}
This theorem excludes a specific class of universal,
state-independent additive complementarity relations involving
$C_q$, $D_Q^{\mathrm{ICO}}$, and $C_{\mathrm{causal}}$.
It does not, however, rule out nonlinear, state-dependent, or
information-theoretic relations involving these quantities or
other process-level quantifiers. Furthermore, the proof relies
on the quantum-switch construction and on the definition of
$C_{\mathrm{causal}}$ as the $l_1$-norm of coherence of the
reduced order qubit; alternative process-level quantifiers may
admit different constraints. The entropic formulation presented
in Sec.~\ref{sec:entr} provides a complementary
information-theoretic perspective that is not restricted to the
assumptions underlying additive trade-off relations.

\medskip

\paragraph*{Explicit realization.}
To illustrate the construction, consider a two-path interferometer with quanton basis $\{|0\rangle,|1\rangle\}$. The quanton is initialized in the balanced superposition
\[
|\psi\rangle_Q = \tfrac{1}{\sqrt{2}} (|0\rangle + |1\rangle),
\quad
|\Psi^{(0)}\rangle = |\psi\rangle_Q \otimes |d_0\rangle,
\]
with the detector initially in the state $|d_0\rangle$.

The which-path interaction is given by
\[
U_A = |0\rangle\langle0|\otimes I_D + |1\rangle\langle1|\otimes X_D,
\]
where $|d_0\rangle$ and $|d_1\rangle$ are orthogonal detector states and $X_D$ denotes a bit-flip operator in this basis. This yields the entangled state
\[
|\Psi_A\rangle := U_A |\Psi^{(0)}\rangle =
\tfrac{1}{\sqrt{2}} \left(
|0\rangle|d_0\rangle + |1\rangle|d_1\rangle
\right),
\]
whose reduced quanton state is $\rho_Q = \tfrac12 I$. Consequently, the spatial coherence vanishes while the path distinguishability is maximal,
\[
C_q = 0, \qquad D_Q^{\mathrm{ICO}} = 1.
\]

Now consider an interference operation acting only on the quanton,
\[
U_B = \left(e^{i\phi_0}|0\rangle\langle0| + e^{i\phi_1}|1\rangle\langle1|\right) \otimes I_D.
\]
Since both $U_A$ and $U_B$ are diagonal in the path basis, they commute, $[U_A,U_B]=0$. Consequently, the two causal orders produce identical states,
\[
|\Psi_{A\prec B}\rangle = |\Psi_{B\prec A}\rangle,
\]
so that $|\langle \Psi_{A\prec B} | \Psi_{B\prec A} \rangle| = 1$. Using Eq.~\eqref{eq:C_causal}, this yields
\[
C_{\mathrm{causal}} = 2\sqrt{p(1-p)},
\]
which attains its maximum value $C_{\mathrm{causal}}=1$ for $p=\tfrac12$.

Thus, the attainable point
$\bigl(C_q+D_Q^{\mathrm{ICO}},C_{\mathrm{causal}}\bigr)=(1,1)$,
rather than subsystem separation alone, establishes the no-go
result. Because $C_{\mathrm{causal}}$ is defined from the reduced
control state, any extension sensitive to the discarded
order--quanton--detector correlations must use joint-state or
process-level information.

\section{Entropic Formulation of Causal-Order Uncertainty}
\label{sec:entr}

In this section, we formulate an information-theoretic complementarity
relation for the order qubit in the presence of indefinite causal order,
using the entropic uncertainty relation with quantum memory~\cite{Berta2010}.
This is distinct from the additive complementarity considered in
Sec.~\ref{sec:no_go} and is not intended as a nonlinear extension of it.
In standard interferometric settings, entropic formulations of
wave--particle duality can be expressed without quantum memory, since
complementary observables act on the same quanton~\cite{Coles2014}.

In the present setting, incompatible measurements act on the control system (the order qubit) associated with alternative causal orders, while the quanton--detector system serves as quantum memory. For the pure switch states considered below, the general relation can be simplified using purity, which yields a bound in terms of the entropy of the reduced order qubit. The corresponding conditional-entropy form retains the joint-state description explicitly.

Consider the pure quantum switch state defined in Eq.~\eqref{eq:switch_state}. Let $Z_O$ and $X_O$ denote two mutually unbiased measurements on the order qubit,
\begin{align}
Z_O &= \{|0\rangle, |1\rangle\}, \\
X_O &= \{|+\rangle, |-\rangle\},
\end{align}
where $|\pm\rangle = (|0\rangle \pm |1\rangle)/\sqrt{2}$. 
The measurement $Z_O$ reveals definite causal order and thus plays a particle-like role, while $X_O$ probes coherent superpositions of causal orders and captures the wave-like behavior of the control system. Since $Z_O$ and $X_O$ are mutually unbiased bases, one has
\[
c = \max_{i,j} |\langle z_i | x_j \rangle|^2 = \tfrac{1}{2},
\]
and hence $\log_2(1/c) = 1$.

Before proceeding, we introduce the conditional entropies appearing in the uncertainty relation. The von Neumann entropy is defined as $H(\rho) := -\operatorname{Tr}(\rho \log_2 \rho)$. For a projective measurement of the order qubit in the $Z_O$ basis, we define
\begin{equation}
H(Z_O|QD) := H(\rho^{Z_O}_{OQD}) - H(\rho_{QD}),
\end{equation}
where $\Pi_z = |z\rangle\langle z|$ are the projectors, and the post-measurement state is
\begin{equation}
\rho^{Z_O}_{OQD}
= \sum_{z}
(\Pi_z \otimes I_{QD}) \,
\rho_{OQD}
(\Pi_z \otimes I_{QD}).
\end{equation}
An analogous definition applies to the $X_O$ basis. This coincides with the standard definition of conditional entropy used in entropic uncertainty relations with quantum memory \cite{Berta2010}.

\begin{corollary}[Entropic uncertainty relation for the order qubit]
For the quantum switch state~\eqref{eq:switch_state}, the conditional entropies associated with measurements $Z_O$ and $X_O$ obey
\begin{equation}
H(Z_O|QD) + H(X_O|QD) \ge 1 - H(O),
\label{eq:state_dependent_bound}
\end{equation}
where $H(O) = H(\rho_O)$ denotes the von Neumann entropy of the reduced order qubit.
\end{corollary}

This relation follows directly from the entropic uncertainty relation with quantum memory~\cite{Berta2010}. Its significance lies in providing an information-theoretic formulation of complementarity for the control system associated with alternative causal orders.

The proof is provided in Appendix~\ref{app:entropic_derivation}.

\medskip

For the reduced order-qubit state introduced in
Sec.~\ref{sec:causal}, with causal coherence given by
Eq.~\eqref{eq:C_causal}, the eigenvalues are
$\lambda_\pm=(1\pm\Delta)/2$, where
\begin{equation}
\Delta = \sqrt{(2p-1)^2 + C_{\mathrm{causal}}^2}.
\end{equation}

Since $\rho_O$ is a density operator, $0 \le \Delta \le 1$, and its von Neumann entropy is
\begin{equation}
H(O) = h_2\!\left(\frac{1+\Delta}{2}\right),
\end{equation}
where $h_2(x) = -x\log_2 x - (1-x)\log_2(1-x)$ denotes the binary entropy.

The entropy vanishes when $\Delta=1$, i.e., when the reduced order qubit is pure. From $\Delta^2 = (2p-1)^2 + 4p(1-p)\,|\langle \Psi_{A\prec B} | \Psi_{B\prec A} \rangle|^2$, this condition implies either $p\in\{0,1\}$ or $|\langle \Psi_{A\prec B} | \Psi_{B\prec A} \rangle|=1$. For $0<p<1$, the latter requires the two causal-order states to coincide up to a global phase. In this regime, the two causal-order states correspond to identical quanton--detector states, so that no which-order information is encoded. As a result, the global state factorizes as $\rho_{OQD} = \rho_O \otimes \rho_{QD}$. Since the global state is pure, this implies that $\rho_{QD}$ is also pure. In particular, in the orthogonal-detector realization used in the no-go construction, the bound is saturated, $C+D_Q=1$. For $p=\tfrac{1}{2}$, one has maximal causal coherence, $C_{\mathrm{causal}}=1$, allowing it to coexist with saturated spatial complementarity, which underlies the no-go theorem. Definite causal order is recovered in the limiting case $p=0$ or $p=1$.

The entropy is maximal, $H(O)=1$, when $\Delta=0$, corresponding to a maximally mixed order qubit.

Substituting this expression for $H(O)$ into Eq.~\eqref{eq:state_dependent_bound} yields the explicit form of the state-dependent bound
\begin{equation}
H(Z_O|QD) + H(X_O|QD)
\ge
1 - h_2\!\left(\frac{1+\Delta}{2}\right).
\label{eq:explicit_state_bound}
\end{equation}

\medskip

\paragraph*{Interpretation.}

The entropic bound reflects the incompatibility of the measurements $Z_O$ and $X_O$, quantified by the overlap parameter $c$. This uncertainty relation is conceptually distinct from the no-go theorem, which follows from the explicit attainability of $\bigl(C_q+D_Q^{\mathrm{ICO}},C_{\mathrm{causal}}\bigr)=(1,1)$.

Causal coherence contributes directly to the entropy of the order qubit through the off-diagonal element $\kappa$, while spatial coherence and path distinguishability enter indirectly through correlations with the quanton--detector system. The spatial quantities $C_q$ and $D_Q^{\mathrm{ICO}}$ do not appear explicitly in the bound because the information they represent---including which-path information---is encoded in these correlations. For the pure switch states considered here, the right-hand side of Eq.~\eqref{eq:explicit_state_bound} is indeed a nonlinear, state-dependent function of $C_{\mathrm{causal}}$ and $p$, as seen explicitly through the entropy $H(O)$ in Eq.~\eqref{eq:state_dependent_bound}. The point is that it is not a relation among $C_q$, $D_Q^{\mathrm{ICO}}$, and $C_{\mathrm{causal}}$ alone. The underlying uncertainty relation contains the conditional entropy $H(O|QD)$ of the joint state, and this joint-state dependence is explicit before the purity substitution. The bound reflects the fundamental limitation on simultaneous knowledge of the causal order ($Z_O$ measurement) and coherence between causal orders ($X_O$ measurement), given the correlations with the spatial degrees of freedom. Thus, while the relation involves spatial information through the quantum memory, it is fundamentally an uncertainty relation for the control system associated with alternative causal orders.

\medskip
\paragraph*{Extension to mixed switch states.}

For a general mixed state $\rho_{OQD}$, the entropic uncertainty relation yields~\cite{Berta2010}
\begin{equation}
H(Z_O|QD) + H(X_O|QD) \ge 1 + H(O|QD),
\end{equation}
which holds without any purity assumption. For the pure switch states considered above, one has $H(O|QD)=-H(O)\leq0$. For a general mixed state, $H(O|QD)$ is not fixed by $H(O)$ alone; it depends on the joint-state correlations and can be positive, zero, or negative.

\medskip

The operational origin of the entropic complementarity relation in Eq.~\eqref{eq:state_dependent_bound} lies in the distinguishability of causal orders, as quantified by the conditional entropy $H(Z_O|QD)$. As shown in Appendix~\ref{app:causal_discrimination}, the Helstrom operator governing optimal discrimination between the two causal orders depends explicitly on the detector operators $|d_i\rangle\langle d_j|$, and hence on their overlaps $\langle d_j|d_i\rangle$. These overlaps are precisely the quantities that determine spatial coherence and path distinguishability in a fixed-order interferometer. Thus, both spatial duality and causal-order discrimination are governed by the same underlying detector correlations.

\section{Conclusion}\label{sec:conc}

We have formulated wave--particle duality in the presence of indefinite causal order within the quantum-switch framework. The control degree of freedom introduces an additional reduced-state quantity---causal coherence---which quantifies coherence of the reduced order qubit in the causal-order basis and is operationally accessible through measurements on that qubit. This quantity is distinct from spatial coherence and should not, by itself, be identified with process-level causal nonseparability.

While tracing out the order qubit recovers the standard duality relation between spatial coherence and path distinguishability, access to the order qubit reveals its reduced coherence in the causal-order basis. For pure switch states, this reduced control coherence satisfies a duality relation with the optimal success probability for discriminating the corresponding ordered output states, within the IDP regime.

Our main result is a no-go theorem showing that no universal state-independent linear additive complementarity relation involving spatial coherence, path distinguishability, and causal coherence exists within the quantum-switch framework. The result follows from an explicit realizable family in which spatial duality is saturated while the reduced control coherence is maximal; it is not inferred from subsystem separation alone.

The failure of universal additive tradeoff relations motivates an alternative entropic description. Complementarity then admits a state-dependent formulation arising from incompatible measurements on the control system associated with alternative causal orders. The strength of this bound is governed by the entropy of the order qubit: it is maximal when the order qubit is pure and becomes trivial when it is maximally mixed. Notably, even when this entropic constraint is strongest, maximal causal coherence can coexist with saturated spatial complementarity. This entropic formulation is not a relation among the same three reduced scalar quantities; rather, it uses conditional entropies that retain correlations with the quanton--detector system as quantum memory.

Thus, complementarity in quantum processes with indefinite causal order reflects the coexistence of spatial and causal structures that are operationally linked through correlations but governed by distinct constraints.

Extending this framework to general process matrices and exploring its implications for quantum information processing tasks involving coherent control of causal order represent natural directions for future work.

\appendix
\section{Convexity of the UQSD distinguishability}
\label{app:uqsd_convexity}

We prove that the optimal success probability of unambiguous quantum state discrimination (UQSD) is convex under classical mixing of detector ensembles.

For a detector ensemble $\mathcal{E}=\{p_i,\sigma_i\}$, where $\sigma_i$ denotes the detector states associated with the interferometric path $i$, the success probability of UQSD for a POVM $\Pi=\{\Pi_i,\Pi_?\}$ is
\begin{equation}
P_{\mathrm{succ}}(\Pi|\mathcal{E})
=
\sum_i p_i\,\operatorname{Tr}(\sigma_i \Pi_i),
\end{equation}
subject to the UQSD constraints
\[
\operatorname{Tr}(\sigma_i \Pi_j)=0 \quad (i\neq j),
\]
and the POVM completeness relation
\[
\sum_i \Pi_i + \Pi_? = I .
\]
The operator $\Pi_?$ corresponds to the inconclusive outcome.

The optimal success probability is obtained by maximizing over all POVMs satisfying these constraints,
\begin{equation}
P_{\mathrm{UQSD}}(\mathcal{E})
=
\max_{\Pi}
\sum_i p_i\,\operatorname{Tr}(\sigma_i \Pi_i).
\end{equation}

Consider two detector ensembles $\mathcal{E}_1$ and $\mathcal{E}_2$. Suppose that $\mathcal{E}_1$ is prepared with probability $p$ and $\mathcal{E}_2$ with probability $1-p$, defining a classical mixture of ensembles. Any POVM that is unambiguous for the mixed ensemble is also unambiguous for each component ensemble, because the wrong-identification probabilities are nonnegative and their convex combination can vanish only when each term vanishes.

For any fixed POVM $\Pi$, the success probability is linear in the ensemble, and one has
\begin{equation}
P_{\mathrm{succ}}(\Pi|\mathcal{E})
=
p\,P_{\mathrm{succ}}(\Pi|\mathcal{E}_1)
+
(1-p)\,P_{\mathrm{succ}}(\Pi|\mathcal{E}_2).
\end{equation}

Since $P_{\mathrm{UQSD}}(\mathcal{E}_k)$ is defined as the maximum over all POVMs, any fixed POVM $\Pi$ satisfies
\[
P_{\mathrm{succ}}(\Pi|\mathcal{E}_k)
\le
P_{\mathrm{UQSD}}(\mathcal{E}_k),
\qquad k\in\{1,2\}.
\]

Therefore,
\begin{equation}
P_{\mathrm{succ}}(\Pi|\mathcal{E})
\le
p\,P_{\mathrm{UQSD}}(\mathcal{E}_1)
+
(1-p)\,P_{\mathrm{UQSD}}(\mathcal{E}_2).
\end{equation}

Since this inequality holds for every POVM $\Pi$, taking the maximum over all POVMs yields
\begin{equation}
P_{\mathrm{UQSD}}(\mathcal{E})
\le
p\,P_{\mathrm{UQSD}}(\mathcal{E}_1)
+
(1-p)\,P_{\mathrm{UQSD}}(\mathcal{E}_2),
\end{equation}
which shows that the optimal UQSD success probability is convex under classical mixing of detector ensembles.

In the main text, the reduced quanton--detector state is given by Eq.~\eqref{eq:convexrho}. Applying this convexity result together with the corresponding fixed-order UQSD upper bounds yields
\begin{equation}
D_Q^{\mathrm{ICO}}
\le
p\,D_Q^{A\prec B}
+
(1-p)\,D_Q^{B\prec A}.
\end{equation}

The formulation of UQSD as a POVM optimization follows the standard framework of quantum state discrimination \cite{Chefles2000,Bergou2007}. Since, for any fixed POVM, the success probability is linear in the ensemble probabilities, and the optimal success probability is obtained by maximizing over all POVMs, it follows that $P_{\mathrm{UQSD}}$ is a convex function of the ensemble probabilities \cite[Sec.~3.2.3]{Boyd2004}.

\section{Post-selected duality relations}
\label{app:post_selected_derivation}

We consider the pure quantum switch state defined in Eq.~\eqref{eq:switch_state}. For a balanced two-path interferometer,
\begin{equation}
|\Psi_X\rangle = \frac{1}{\sqrt{2}} \sum_{i=0}^1 |i\rangle |d_i^X\rangle,
\quad X\in\{A\prec B,\,B\prec A\},
\end{equation}
with normalized detector states $\langle d_i^X|d_i^X\rangle=1$.

Projecting the order qubit onto $|\pm\rangle=(|0\rangle\pm|1\rangle)/\sqrt{2}$ gives
\begin{equation}
|\Psi^{(\pm)}\rangle
=
\frac{1}{2\sqrt{\mathcal{N}_\pm}}
\sum_{i=0}^1 |i\rangle |d_i^{(\pm)}\rangle,
\end{equation}
where
\begin{equation}
|d_i^{(\pm)}\rangle
=
\sqrt{p}\,|d_i^{A\prec B}\rangle
\pm
e^{i\theta}\sqrt{1-p}\,|d_i^{B\prec A}\rangle,
\end{equation}
and $\mathcal{N}_\pm$ is the normalization factor,
\begin{equation}
\mathcal{N}_\pm
=
\frac{1}{2}
\left[
1 \pm \sqrt{p(1-p)}\,\mathrm{Re}\!\left(
e^{i\theta}(\Gamma_{00}+\Gamma_{11})
\right)
\right],
\end{equation}
with $\Gamma_{ij}=\langle d_i^{A\prec B}|d_j^{B\prec A}\rangle$.

We assume $\mathcal{N}_\pm > 0$ when discussing the corresponding post-selected conditional states. For identical causal-order states with $p = \tfrac{1}{2}$, one has $\mathcal{N}_- = 0$, corresponding to a zero-probability post-selected outcome. Hence the corresponding conditional state is not physically realized and is not defined at this parameter point. The no-go construction in Sec.~\ref{sec:no_go} is unaffected, since it does not condition on an order-qubit measurement outcome.

Tracing over the detector subsystem yields the reduced quanton state
\begin{equation}
\rho_Q^{(\pm)}
=
\frac{1}{4\mathcal{N}_\pm}
\sum_{i,j}
|i\rangle\langle j|\,
\langle d_j^{(\pm)}|d_i^{(\pm)}\rangle.
\end{equation}

The off-diagonal element is given by
\begin{equation}
\gamma_\pm
=
\frac{1}{4\mathcal{N}_\pm}
\left[
C_{\mathrm{mix}}
\pm
\sqrt{p(1-p)}
\big(
e^{i\theta}\Gamma_{10}
+
e^{-i\theta}\Gamma_{01}^*
\big)
\right],
\end{equation}
where $C_{\mathrm{mix}} = p\,\gamma_{A\prec B} + (1-p)\,\gamma_{B\prec A}$, with $\gamma_X = \langle d_1^{X} | d_0^{X} \rangle$ for $X \in \{A\prec B,\, B\prec A\}$.

\medskip

The reduced state takes the form with equal diagonal entries only under the symmetry condition
\begin{equation}
\mathrm{Re}(e^{i\theta}\Gamma_{00}) = \mathrm{Re}(e^{i\theta}\Gamma_{11}).
\end{equation}
Without this symmetry, the diagonal entries may be unequal, and the general form of the reduced quanton state is
\begin{equation}
\rho_Q^{(\pm)} = 
\begin{pmatrix}
p_\pm & \gamma_\pm \\
\gamma_\pm^* & 1-p_\pm
\end{pmatrix},
\end{equation}
with $p_\pm \neq 1/2$ in general. In this case, the relation $C^{(\pm)} + D^{(\pm)} = 1$ is replaced by the inequality $C^{(\pm)} + D^{(\pm)} \le 1$, following the standard duality relation for mixed states~\cite{Bera2015}.

In the symmetric case satisfying the above condition, the reduced state takes the form
\begin{equation}
\rho_Q^{(\pm)}
=
\begin{pmatrix}
\frac{1}{2} & \gamma_\pm \\
\gamma_\pm^* & \frac{1}{2}
\end{pmatrix}.
\end{equation}

The $l_1$-norm of coherence is $C^{(\pm)} = 2|\gamma_\pm|$. In this symmetric two-path case, the distinguishability is given by
\begin{equation}
D^{(\pm)} = 1 - 2|\gamma_\pm|.
\end{equation}
This coincides with the optimal success probability in UQSD for two states and satisfies the coherence--distinguishability duality relation~\cite{Bera2015},
\begin{equation}
C^{(\pm)} + D^{(\pm)} = 1.
\end{equation}

The post-selected interference visibility is given by $2\lvert \gamma_\pm \rvert$, which may be less than the maximum possible given the distinguishability when the symmetry condition is not satisfied.

\section{Derivation of the entropic complementarity relation}
\label{app:entropic_derivation}

We derive Eq.~\eqref{eq:state_dependent_bound} using the entropic uncertainty relation with quantum memory \cite{Berta2010}.

For two measurements $R$ and $S$ on system $O$ in the presence of quantum memory $QD$, the entropic uncertainty relation of Berta et al.~\cite{Berta2010} gives
\begin{equation}
H(R|QD) + H(S|QD)
\ge \log_2 \frac{1}{c} + H(O|QD),
\end{equation}
where $c=\max_{i,j}|\langle r_i|s_j\rangle|^2$ quantifies the overlap between the measurement bases.

For mutually unbiased bases $Z_O$ and $X_O$, one has $c=\tfrac12$, and hence
\begin{equation}
H(Z_O|QD) + H(X_O|QD) \ge 1 + H(O|QD).
\end{equation}

For a pure global state $\rho_{OQD}$, the conditional entropy satisfies
\begin{equation}
H(O|QD) = -H(O),
\end{equation}
since $H(O|QD) = H(\rho_{OQD}) - H(\rho_{QD})$, $H(\rho_{OQD})=0$, and for a pure bipartite state one has $H(\rho_{QD}) = H(\rho_O)$.

Substituting this into the above inequality yields
\begin{equation}
H(Z_O|QD) + H(X_O|QD) \ge 1 - H(O),
\end{equation}
as stated in Eq.~\eqref{eq:state_dependent_bound}.

\section{Detector correlations and causal-order discrimination}
\label{app:causal_discrimination}

To elucidate the operational origin of the entropic complementarity relation in Eq.~\eqref{eq:state_dependent_bound}, we analyze how detector correlations determine the distinguishability of causal orders.

\begin{lemma}[Detector correlations underlying causal-order discrimination]
Let $U_A$ denote the which-path interaction
\begin{equation}
U_A:\quad \ket{\psi_i}\ket{d_0} \longrightarrow
\ket{\psi_i}\ket{d_i},
\end{equation}
and let $U_B = U_Q \otimes I_D$ act nontrivially only on the quanton. Then the Helstrom operator for discriminating the two causal orders,
\begin{equation}
\Delta := p\rho_{A\prec B}-(1-p)\rho_{B\prec A},
\end{equation}
depends explicitly on the detector overlaps $\langle d_j|d_i\rangle$. Consequently, the optimal discrimination probability between causal orders is governed by the same detector correlations that determine spatial coherence and path distinguishability in a fixed-order interferometer.
\end{lemma}

\begin{proof}
The two definite-order processes generated by $U_A$ and $U_B$ are
\begin{align}
\rho_{A\prec B} & = U_B U_A \rho^{(0)}_{QD} U_A^\dagger U_B^\dagger,\\
\rho_{B\prec A} & = U_A U_B \rho^{(0)}_{QD} U_B^\dagger U_A^\dagger .
\end{align}

After the which-path interaction $U_A$, the joint quanton--detector state takes the form
\begin{equation}
\rho_{QD} =
\sum_{i,j} \sqrt{p_i p_j}\, e^{i(\phi_i-\phi_j)}
|\psi_i\rangle\langle\psi_j|
\otimes
|d_i\rangle\langle d_j|,
\end{equation}
which coincides with the fixed-order state in Eq.~\eqref{eq:fixed_dop}.

Since $U_B = U_Q \otimes I_D$ acts trivially on the detector Hilbert space, the operators $|d_i\rangle\langle d_j|$ remain invariant under $U_B$ and $U_B^\dagger$. Consequently, both $\rho_{A\prec B}$ and $\rho_{B\prec A}$ retain tensor factors of the form
\[
|\psi_i\rangle\langle\psi_j|\otimes|d_i\rangle\langle d_j|.
\]

Therefore, the Helstrom operator
\[
\Delta = p\rho_{A\prec B}-(1-p)\rho_{B\prec A}
\]
inherits these detector operators. Tracing over the detector yields
\[
\operatorname{Tr}_D\!\left(|d_i\rangle\langle d_j|\right)
= \langle d_j|d_i\rangle,
\]
so that $\Delta$ depends explicitly on the overlaps $\langle d_j|d_i\rangle$.

Since the optimal discrimination probability is determined by the trace norm $\|\Delta\|_1$, it follows that causal-order distinguishability is governed by the same overlaps $\langle d_j|d_i\rangle$ that determine spatial coherence and path distinguishability in the fixed-order interferometer.
\end{proof}

\begin{corollary}[Operational interpretation]
Let $\rho_{QD}=\operatorname{Tr}_O(|\Psi_{\mathrm{tot}}\rangle\langle\Psi_{\mathrm{tot}}|)$ be the reduced quanton--detector state defined in Eq.~\eqref{eq:convexrho}. A projective measurement of $Z_O$ on the order qubit prepares the classical--quantum ensemble $\{\, p,\, \rho_{A\prec B};\; 1-p,\, \rho_{B\prec A} \,\}$ on $QD$.

The optimal probability for minimum-error discrimination of the causal order, given access to $QD$, is therefore given by the Helstrom bound~\cite{Helstrom1976,Bergou2010},
\begin{equation}
P_{\mathrm{guess}}(Z_O|QD)
=
\frac12
\left(
1 +
\left\|
p\rho_{A\prec B}
-
(1-p)\rho_{B\prec A}
\right\|_1
\right),
\label{eq:helstrom_operator}
\end{equation}
where $\|\cdot\|_1$ denotes the trace norm.
\end{corollary}
By the preceding lemma, the Helstrom operator $p\rho_{A\prec B}-(1-p)\rho_{B\prec A}$ depends explicitly on the detector overlaps $\langle d_j|d_i\rangle$, thereby linking optimal discrimination of causal orders to the same detector correlation structure that governs fixed-order wave--particle duality.

Accordingly, discrimination of the two causal orders is formulated as a minimum-error state discrimination problem. This is operationally distinct from the path distinguishability defined earlier via unambiguous quantum state discrimination (UQSD), as these correspond to different notions of distinguishability. Nevertheless, both are governed by the same detector overlaps $\langle d_j|d_i\rangle$ that encode correlations between the quanton and detector.

Together with Eq.~\eqref{eq:explicit_state_bound}, this establishes a genuinely \emph{causal} form of complementarity: increased distinguishability of the causal order (corresponding to smaller $H(Z_O|QD)$) is accompanied by increased uncertainty in measurements performed in a superposition basis. Unlike standard wave--particle duality, this complementarity is not algebraic but entropic in nature, arising from the incompatibility of measurements on the control system associated with alternative temporal structures. This demonstrates that complementarity in indefinite causal order does not impose a joint constraint on spatial and causal observables, but instead manifests as an information-theoretic constraint governing incompatible measurements of the control system.
 
 \input{revised_2oct.bbl}

\end{document}

%% file: revised_2oct.bbl
%

%% file: revised_2oct.bbl
\begin{thebibliography}{54}%
\makeatletter
\providecommand \@ifxundefined [1]{%
 \@ifx{#1\undefined}
}%
\providecommand \@ifnum [1]{%
 \ifnum #1\expandafter \@firstoftwo
 \else \expandafter \@secondoftwo
 \fi
}%
\providecommand \@ifx [1]{%
 \ifx #1\expandafter \@firstoftwo
 \else \expandafter \@secondoftwo
 \fi
}%
\providecommand \natexlab [1]{#1}%
\providecommand \enquote  [1]{``#1''}%
\providecommand \bibnamefont  [1]{#1}%
\providecommand \bibfnamefont [1]{#1}%
\providecommand \citenamefont [1]{#1}%
\providecommand \href@noop [0]{\@secondoftwo}%
\providecommand \href [0]{\begingroup \@sanitize@url \@href}%
\providecommand \@href[1]{\@@startlink{#1}\@@href}%
\providecommand \@@href[1]{\endgroup#1\@@endlink}%
\providecommand \@sanitize@url [0]{\catcode `\\12\catcode `\$12\catcode
  `\&12\catcode `\#12\catcode `\^12\catcode `\_12\catcode `\%12\relax}%
\providecommand \@@startlink[1]{}%
\providecommand \@@endlink[0]{}%
\providecommand \url  [0]{\begingroup\@sanitize@url \@url }%
\providecommand \@url [1]{\endgroup\@href {#1}{\urlprefix }}%
\providecommand \urlprefix  [0]{URL }%
\providecommand \Eprint [0]{\href }%
\providecommand \doibase [0]{https://doi.org/}%
\providecommand \selectlanguage [0]{\@gobble}%
\providecommand \bibinfo  [0]{\@secondoftwo}%
\providecommand \bibfield  [0]{\@secondoftwo}%
\providecommand \translation [1]{[#1]}%
\providecommand \BibitemOpen [0]{}%
\providecommand \bibitemStop [0]{}%
\providecommand \bibitemNoStop [0]{.\EOS\space}%
\providecommand \EOS [0]{\spacefactor3000\relax}%
\providecommand \BibitemShut  [1]{\csname bibitem#1\endcsname}%
\let\auto@bib@innerbib\@empty
\bibitem [{\citenamefont {Starke}\ \emph {et~al.}(2026)\citenamefont {Starke},
  \citenamefont {Maziero}, \citenamefont {Basso},\ and\ \citenamefont
  {Qureshi}}]{starke2026}%
  \BibitemOpen
  \bibfield  {author} {\bibinfo {author} {\bibfnamefont {D.~S.}\ \bibnamefont
  {Starke}}, \bibinfo {author} {\bibfnamefont {J.}~\bibnamefont {Maziero}},
  \bibinfo {author} {\bibfnamefont {M.~L.~W.}\ \bibnamefont {Basso}},\ and\
  \bibinfo {author} {\bibfnamefont {T.}~\bibnamefont {Qureshi}},\ }\href
  {https://arxiv.org/abs/2605.26375} {\bibinfo {title} {Bohr's
  complementarity}} (\bibinfo {year} {2026}),\ \Eprint
  {https://arxiv.org/abs/2605.26375} {arXiv:2605.26375 [quant-ph]} \BibitemShut
  {NoStop}%
\bibitem [{\citenamefont {Wootters}\ and\ \citenamefont
  {Zurek}(1979)}]{WoottersZurek}%
  \BibitemOpen
  \bibfield  {author} {\bibinfo {author} {\bibfnamefont {W.~K.}\ \bibnamefont
  {Wootters}}\ and\ \bibinfo {author} {\bibfnamefont {W.~H.}\ \bibnamefont
  {Zurek}},\ }\bibfield  {title} {\bibinfo {title} {Complementarity in the
  double-slit experiment: Quantum nonseparability and a quantitative statement
  of bohr's principle},\ }\href {https://doi.org/10.1103/PhysRevD.19.473}
  {\bibfield  {journal} {\bibinfo  {journal} {Phys. Rev. D}\ }\textbf {\bibinfo
  {volume} {19}},\ \bibinfo {pages} {473} (\bibinfo {year} {1979})}\BibitemShut
  {NoStop}%
\bibitem [{\citenamefont {Greenberger}\ and\ \citenamefont
  {Yasin}(1988)}]{GreenbergerYasin}%
  \BibitemOpen
  \bibfield  {author} {\bibinfo {author} {\bibfnamefont {D.~M.}\ \bibnamefont
  {Greenberger}}\ and\ \bibinfo {author} {\bibfnamefont {A.}~\bibnamefont
  {Yasin}},\ }\bibfield  {title} {\bibinfo {title} {Simultaneous wave and
  particle knowledge in a neutron interferometer},\ }\href
  {https://doi.org/https://doi.org/10.1016/0375-9601(88)90114-4} {\bibfield
  {journal} {\bibinfo  {journal} {Physics Letters A}\ }\textbf {\bibinfo
  {volume} {128}},\ \bibinfo {pages} {391} (\bibinfo {year}
  {1988})}\BibitemShut {NoStop}%
\bibitem [{\citenamefont {Englert}(1996)}]{Englert}%
  \BibitemOpen
  \bibfield  {author} {\bibinfo {author} {\bibfnamefont {B.-G.}\ \bibnamefont
  {Englert}},\ }\bibfield  {title} {\bibinfo {title} {Fringe visibility and
  which-way information: An inequality},\ }\href
  {https://doi.org/10.1103/PhysRevLett.77.2154} {\bibfield  {journal} {\bibinfo
   {journal} {Phys. Rev. Lett.}\ }\textbf {\bibinfo {volume} {77}},\ \bibinfo
  {pages} {2154} (\bibinfo {year} {1996})}\BibitemShut {NoStop}%
\bibitem [{\citenamefont {D\"urr}(2001)}]{Duerr2001}%
  \BibitemOpen
  \bibfield  {author} {\bibinfo {author} {\bibfnamefont {S.}~\bibnamefont
  {D\"urr}},\ }\bibfield  {title} {\bibinfo {title} {Quantitative wave-particle
  duality in multibeam interferometers},\ }\href
  {https://doi.org/10.1103/PhysRevA.64.042113} {\bibfield  {journal} {\bibinfo
  {journal} {Phys. Rev. A}\ }\textbf {\bibinfo {volume} {64}},\ \bibinfo
  {pages} {042113} (\bibinfo {year} {2001})}\BibitemShut {NoStop}%
\bibitem [{\citenamefont {Asad~Siddiqui}\ and\ \citenamefont
  {Qureshi}(2015)}]{SiddiquiQureshi2015}%
  \BibitemOpen
  \bibfield  {author} {\bibinfo {author} {\bibfnamefont {M.}~\bibnamefont
  {Asad~Siddiqui}}\ and\ \bibinfo {author} {\bibfnamefont {T.}~\bibnamefont
  {Qureshi}},\ }\bibfield  {title} {\bibinfo {title} {Three-slit interference:
  A duality relation},\ }\href {https://doi.org/10.1093/ptep/ptv112} {\bibfield
   {journal} {\bibinfo  {journal} {Progress of Theoretical and Experimental
  Physics}\ }\textbf {\bibinfo {volume} {2015}},\ \bibinfo {pages} {083A02}
  (\bibinfo {year} {2015})}\BibitemShut {NoStop}%
\bibitem [{\citenamefont {Bera}\ \emph {et~al.}(2015)\citenamefont {Bera},
  \citenamefont {Qureshi}, \citenamefont {Siddiqui},\ and\ \citenamefont
  {Pati}}]{Bera2015}%
  \BibitemOpen
  \bibfield  {author} {\bibinfo {author} {\bibfnamefont {M.~N.}\ \bibnamefont
  {Bera}}, \bibinfo {author} {\bibfnamefont {T.}~\bibnamefont {Qureshi}},
  \bibinfo {author} {\bibfnamefont {M.~A.}\ \bibnamefont {Siddiqui}},\ and\
  \bibinfo {author} {\bibfnamefont {A.~K.}\ \bibnamefont {Pati}},\ }\bibfield
  {title} {\bibinfo {title} {Duality of quantum coherence and path
  distinguishability},\ }\href {https://doi.org/10.1103/PhysRevA.92.012118}
  {\bibfield  {journal} {\bibinfo  {journal} {Phys. Rev. A}\ }\textbf {\bibinfo
  {volume} {92}},\ \bibinfo {pages} {012118} (\bibinfo {year}
  {2015})}\BibitemShut {NoStop}%
\bibitem [{\citenamefont {Bagan}\ \emph {et~al.}(2016)\citenamefont {Bagan},
  \citenamefont {Bergou}, \citenamefont {Cottrell},\ and\ \citenamefont
  {Hillery}}]{Bagan2016}%
  \BibitemOpen
  \bibfield  {author} {\bibinfo {author} {\bibfnamefont {E.}~\bibnamefont
  {Bagan}}, \bibinfo {author} {\bibfnamefont {J.~A.}\ \bibnamefont {Bergou}},
  \bibinfo {author} {\bibfnamefont {S.~S.}\ \bibnamefont {Cottrell}},\ and\
  \bibinfo {author} {\bibfnamefont {M.}~\bibnamefont {Hillery}},\ }\bibfield
  {title} {\bibinfo {title} {Relations between coherence and path
  information},\ }\href {https://doi.org/10.1103/PhysRevLett.116.160406}
  {\bibfield  {journal} {\bibinfo  {journal} {Phys. Rev. Lett.}\ }\textbf
  {\bibinfo {volume} {116}},\ \bibinfo {pages} {160406} (\bibinfo {year}
  {2016})}\BibitemShut {NoStop}%
\bibitem [{\citenamefont {Qureshi}\ and\ \citenamefont
  {Siddiqui}(2017)}]{QureshiSiddiqui2017}%
  \BibitemOpen
  \bibfield  {author} {\bibinfo {author} {\bibfnamefont {T.}~\bibnamefont
  {Qureshi}}\ and\ \bibinfo {author} {\bibfnamefont {M.~A.}\ \bibnamefont
  {Siddiqui}},\ }\bibfield  {title} {\bibinfo {title} {Wave–particle duality
  in n-path interference},\ }\href
  {https://doi.org/https://doi.org/10.1016/j.aop.2017.08.015} {\bibfield
  {journal} {\bibinfo  {journal} {Annals of Physics}\ }\textbf {\bibinfo
  {volume} {385}},\ \bibinfo {pages} {598} (\bibinfo {year}
  {2017})}\BibitemShut {NoStop}%
\bibitem [{\citenamefont {Qureshi}(2019)}]{Qureshi2019}%
  \BibitemOpen
  \bibfield  {author} {\bibinfo {author} {\bibfnamefont {T.}~\bibnamefont
  {Qureshi}},\ }\bibfield  {title} {\bibinfo {title} {Interference visibility
  and wave-particle duality in multipath interference},\ }\href
  {https://doi.org/10.1103/PhysRevA.100.042105} {\bibfield  {journal} {\bibinfo
   {journal} {Phys. Rev. A}\ }\textbf {\bibinfo {volume} {100}},\ \bibinfo
  {pages} {042105} (\bibinfo {year} {2019})}\BibitemShut {NoStop}%
\bibitem [{\citenamefont {Siddiqui}\ and\ \citenamefont
  {Qureshi}(2021)}]{SiddiquiQureshi2021}%
  \BibitemOpen
  \bibfield  {author} {\bibinfo {author} {\bibfnamefont {M.~A.}\ \bibnamefont
  {Siddiqui}}\ and\ \bibinfo {author} {\bibfnamefont {T.}~\bibnamefont
  {Qureshi}},\ }\bibfield  {title} {\bibinfo {title} {Multipath wave-particle
  duality with a path detector in a quantum superposition},\ }\href
  {https://doi.org/10.1103/PhysRevA.103.022219} {\bibfield  {journal} {\bibinfo
   {journal} {Phys. Rev. A}\ }\textbf {\bibinfo {volume} {103}},\ \bibinfo
  {pages} {022219} (\bibinfo {year} {2021})}\BibitemShut {NoStop}%
\bibitem [{\citenamefont {Baumgratz}\ \emph {et~al.}(2014)\citenamefont
  {Baumgratz}, \citenamefont {Cramer},\ and\ \citenamefont
  {Plenio}}]{Baumgratz2014}%
  \BibitemOpen
  \bibfield  {author} {\bibinfo {author} {\bibfnamefont {T.}~\bibnamefont
  {Baumgratz}}, \bibinfo {author} {\bibfnamefont {M.}~\bibnamefont {Cramer}},\
  and\ \bibinfo {author} {\bibfnamefont {M.~B.}\ \bibnamefont {Plenio}},\
  }\bibfield  {title} {\bibinfo {title} {Quantifying coherence},\ }\href
  {https://doi.org/10.1103/PhysRevLett.113.140401} {\bibfield  {journal}
  {\bibinfo  {journal} {Phys. Rev. Lett.}\ }\textbf {\bibinfo {volume} {113}},\
  \bibinfo {pages} {140401} (\bibinfo {year} {2014})}\BibitemShut {NoStop}%
\bibitem [{\citenamefont {Streltsov}\ \emph {et~al.}(2017)\citenamefont
  {Streltsov}, \citenamefont {Adesso},\ and\ \citenamefont
  {Plenio}}]{Streltsov2017}%
  \BibitemOpen
  \bibfield  {author} {\bibinfo {author} {\bibfnamefont {A.}~\bibnamefont
  {Streltsov}}, \bibinfo {author} {\bibfnamefont {G.}~\bibnamefont {Adesso}},\
  and\ \bibinfo {author} {\bibfnamefont {M.~B.}\ \bibnamefont {Plenio}},\
  }\bibfield  {title} {\bibinfo {title} {Colloquium: Quantum coherence as a
  resource},\ }\href {https://doi.org/10.1103/RevModPhys.89.041003} {\bibfield
  {journal} {\bibinfo  {journal} {Rev. Mod. Phys.}\ }\textbf {\bibinfo {volume}
  {89}},\ \bibinfo {pages} {041003} (\bibinfo {year} {2017})}\BibitemShut
  {NoStop}%
\bibitem [{\citenamefont {Jakob}\ and\ \citenamefont
  {Bergou}(2007)}]{JakobBergou2007}%
  \BibitemOpen
  \bibfield  {author} {\bibinfo {author} {\bibfnamefont {M.}~\bibnamefont
  {Jakob}}\ and\ \bibinfo {author} {\bibfnamefont {J.~A.}\ \bibnamefont
  {Bergou}},\ }\bibfield  {title} {\bibinfo {title} {Complementarity and
  entanglement in bipartite qudit systems},\ }\href
  {https://doi.org/10.1103/PhysRevA.76.052107} {\bibfield  {journal} {\bibinfo
  {journal} {Phys. Rev. A}\ }\textbf {\bibinfo {volume} {76}},\ \bibinfo
  {pages} {052107} (\bibinfo {year} {2007})}\BibitemShut {NoStop}%
\bibitem [{\citenamefont {Bu}\ \emph {et~al.}(2018)\citenamefont {Bu},
  \citenamefont {Li}, \citenamefont {Wu},\ and\ \citenamefont {Fei}}]{Bu2018}%
  \BibitemOpen
  \bibfield  {author} {\bibinfo {author} {\bibfnamefont {K.}~\bibnamefont
  {Bu}}, \bibinfo {author} {\bibfnamefont {L.}~\bibnamefont {Li}}, \bibinfo
  {author} {\bibfnamefont {J.}~\bibnamefont {Wu}},\ and\ \bibinfo {author}
  {\bibfnamefont {S.-M.}\ \bibnamefont {Fei}},\ }\bibfield  {title} {\bibinfo
  {title} {Duality relation between coherence and path information in the
  presence of quantum memory},\ }\href
  {https://doi.org/10.1088/1751-8121/aa9b4f} {\bibfield  {journal} {\bibinfo
  {journal} {Journal of Physics A: Mathematical and Theoretical}\ }\textbf
  {\bibinfo {volume} {51}},\ \bibinfo {pages} {085304} (\bibinfo {year}
  {2018})}\BibitemShut {NoStop}%
\bibitem [{\citenamefont {Roy}\ \emph {et~al.}(2022)\citenamefont {Roy},
  \citenamefont {Pathania}, \citenamefont {Chandra}, \citenamefont
  {Panigrahi},\ and\ \citenamefont {Qureshi}}]{Roy2022}%
  \BibitemOpen
  \bibfield  {author} {\bibinfo {author} {\bibfnamefont {A.~K.}\ \bibnamefont
  {Roy}}, \bibinfo {author} {\bibfnamefont {N.}~\bibnamefont {Pathania}},
  \bibinfo {author} {\bibfnamefont {N.~K.}\ \bibnamefont {Chandra}}, \bibinfo
  {author} {\bibfnamefont {P.~K.}\ \bibnamefont {Panigrahi}},\ and\ \bibinfo
  {author} {\bibfnamefont {T.}~\bibnamefont {Qureshi}},\ }\bibfield  {title}
  {\bibinfo {title} {Coherence, path predictability, and $i$ concurrence: A
  triality},\ }\href {https://doi.org/10.1103/PhysRevA.105.032209} {\bibfield
  {journal} {\bibinfo  {journal} {Phys. Rev. A}\ }\textbf {\bibinfo {volume}
  {105}},\ \bibinfo {pages} {032209} (\bibinfo {year} {2022})}\BibitemShut
  {NoStop}%
\bibitem [{\citenamefont {Sun}\ \emph {et~al.}(2025)\citenamefont {Sun},
  \citenamefont {Zhao},\ and\ \citenamefont {Li}}]{Sun2025}%
  \BibitemOpen
  \bibfield  {author} {\bibinfo {author} {\bibfnamefont {Y.}~\bibnamefont
  {Sun}}, \bibinfo {author} {\bibfnamefont {M.-J.}\ \bibnamefont {Zhao}},\ and\
  \bibinfo {author} {\bibfnamefont {P.-T.}\ \bibnamefont {Li}},\ }\bibfield
  {title} {\bibinfo {title} {Complementarity relations in a multipath
  interferometer with quantum memory},\ }\href
  {https://doi.org/10.1103/ms72-2mzb} {\bibfield  {journal} {\bibinfo
  {journal} {Phys. Rev. A}\ }\textbf {\bibinfo {volume} {112}},\ \bibinfo
  {pages} {012427} (\bibinfo {year} {2025})}\BibitemShut {NoStop}%
\bibitem [{\citenamefont {Banaszek}\ \emph {et~al.}(2013)\citenamefont
  {Banaszek}, \citenamefont {Horodecki}, \citenamefont {Karpi\'{n}ski},\ and\
  \citenamefont {Radzewicz}}]{Banaszek2013}%
  \BibitemOpen
  \bibfield  {author} {\bibinfo {author} {\bibfnamefont {K.}~\bibnamefont
  {Banaszek}}, \bibinfo {author} {\bibfnamefont {P.}~\bibnamefont {Horodecki}},
  \bibinfo {author} {\bibfnamefont {M.}~\bibnamefont {Karpi\'{n}ski}},\ and\
  \bibinfo {author} {\bibfnamefont {C.}~\bibnamefont {Radzewicz}},\ }\bibfield
  {title} {\bibinfo {title} {Quantum mechanical which-way experiment with an
  internal degree of freedom},\ }\href {https://doi.org/10.1038/ncomms3594}
  {\bibfield  {journal} {\bibinfo  {journal} {Nature Communications}\ }\textbf
  {\bibinfo {volume} {4}},\ \bibinfo {pages} {2594} (\bibinfo {year}
  {2013})}\BibitemShut {NoStop}%
\bibitem [{\citenamefont {Sharma}\ \emph {et~al.}(2020)\citenamefont {Sharma},
  \citenamefont {Siddiqui}, \citenamefont {Mal}, \citenamefont {Sazim},\ and\
  \citenamefont {Sen(De)}}]{Sharma2020}%
  \BibitemOpen
  \bibfield  {author} {\bibinfo {author} {\bibfnamefont {G.}~\bibnamefont
  {Sharma}}, \bibinfo {author} {\bibfnamefont {M.~A.}\ \bibnamefont
  {Siddiqui}}, \bibinfo {author} {\bibfnamefont {S.}~\bibnamefont {Mal}},
  \bibinfo {author} {\bibfnamefont {S.}~\bibnamefont {Sazim}},\ and\ \bibinfo
  {author} {\bibfnamefont {A.}~\bibnamefont {Sen(De)}},\ }\bibfield  {title}
  {\bibinfo {title} {Robustness of interferometric complementarity under
  decoherence},\ }\href
  {https://doi.org/https://doi.org/10.1016/j.physleta.2020.126297} {\bibfield
  {journal} {\bibinfo  {journal} {Physics Letters A}\ }\textbf {\bibinfo
  {volume} {384}},\ \bibinfo {pages} {126297} (\bibinfo {year}
  {2020})}\BibitemShut {NoStop}%
\bibitem [{\citenamefont {Oreshkov}\ \emph {et~al.}(2012)\citenamefont
  {Oreshkov}, \citenamefont {Costa},\ and\ \citenamefont
  {Brukner}}]{Oreshkov2012}%
  \BibitemOpen
  \bibfield  {author} {\bibinfo {author} {\bibfnamefont {O.}~\bibnamefont
  {Oreshkov}}, \bibinfo {author} {\bibfnamefont {F.}~\bibnamefont {Costa}},\
  and\ \bibinfo {author} {\bibfnamefont {{\v{C}}.}~\bibnamefont {Brukner}},\
  }\bibfield  {title} {\bibinfo {title} {Quantum correlations with no causal
  order},\ }\href {https://doi.org/10.1038/ncomms2076} {\bibfield  {journal}
  {\bibinfo  {journal} {Nature Communications}\ }\textbf {\bibinfo {volume}
  {3}},\ \bibinfo {pages} {1092} (\bibinfo {year} {2012})}\BibitemShut
  {NoStop}%
\bibitem [{\citenamefont {Chiribella}\ \emph {et~al.}(2013)\citenamefont
  {Chiribella}, \citenamefont {D'Ariano}, \citenamefont {Perinotti},\ and\
  \citenamefont {Valiron}}]{Chiribella2013}%
  \BibitemOpen
  \bibfield  {author} {\bibinfo {author} {\bibfnamefont {G.}~\bibnamefont
  {Chiribella}}, \bibinfo {author} {\bibfnamefont {G.~M.}\ \bibnamefont
  {D'Ariano}}, \bibinfo {author} {\bibfnamefont {P.}~\bibnamefont
  {Perinotti}},\ and\ \bibinfo {author} {\bibfnamefont {B.}~\bibnamefont
  {Valiron}},\ }\bibfield  {title} {\bibinfo {title} {Quantum computations
  without definite causal structure},\ }\href
  {https://doi.org/10.1103/PhysRevA.88.022318} {\bibfield  {journal} {\bibinfo
  {journal} {Phys. Rev. A}\ }\textbf {\bibinfo {volume} {88}},\ \bibinfo
  {pages} {022318} (\bibinfo {year} {2013})}\BibitemShut {NoStop}%
\bibitem [{\citenamefont {Brukner}(2014)}]{Brukner2014}%
  \BibitemOpen
  \bibfield  {author} {\bibinfo {author} {\bibfnamefont {{\v{C}}.}~\bibnamefont
  {Brukner}},\ }\bibfield  {title} {\bibinfo {title} {Quantum causality},\
  }\href {https://doi.org/10.1038/nphys2930} {\bibfield  {journal} {\bibinfo
  {journal} {Nature Physics}\ }\textbf {\bibinfo {volume} {10}},\ \bibinfo
  {pages} {259} (\bibinfo {year} {2014})}\BibitemShut {NoStop}%
\bibitem [{\citenamefont {Procopio}\ \emph {et~al.}(2015)\citenamefont
  {Procopio}, \citenamefont {Moqanaki}, \citenamefont {Ara{\'u}jo},
  \citenamefont {Costa}, \citenamefont {Calafell}, \citenamefont {Dowd},
  \citenamefont {Hamel}, \citenamefont {Rozema}, \citenamefont {Brukner},\ and\
  \citenamefont {Walther}}]{Procopio2015}%
  \BibitemOpen
  \bibfield  {author} {\bibinfo {author} {\bibfnamefont {L.~M.}\ \bibnamefont
  {Procopio}}, \bibinfo {author} {\bibfnamefont {A.}~\bibnamefont {Moqanaki}},
  \bibinfo {author} {\bibfnamefont {M.}~\bibnamefont {Ara{\'u}jo}}, \bibinfo
  {author} {\bibfnamefont {F.}~\bibnamefont {Costa}}, \bibinfo {author}
  {\bibfnamefont {I.~A.}\ \bibnamefont {Calafell}}, \bibinfo {author}
  {\bibfnamefont {E.~G.}\ \bibnamefont {Dowd}}, \bibinfo {author}
  {\bibfnamefont {D.~R.}\ \bibnamefont {Hamel}}, \bibinfo {author}
  {\bibfnamefont {L.~A.}\ \bibnamefont {Rozema}}, \bibinfo {author}
  {\bibfnamefont {{\v{C}}.}~\bibnamefont {Brukner}},\ and\ \bibinfo {author}
  {\bibfnamefont {P.}~\bibnamefont {Walther}},\ }\bibfield  {title} {\bibinfo
  {title} {Experimental superposition of orders of quantum gates},\ }\href
  {https://doi.org/10.1038/ncomms8913} {\bibfield  {journal} {\bibinfo
  {journal} {Nature Communications}\ }\textbf {\bibinfo {volume} {6}},\
  \bibinfo {pages} {7913} (\bibinfo {year} {2015})}\BibitemShut {NoStop}%
\bibitem [{\citenamefont {Rubino}\ \emph {et~al.}(2017)\citenamefont {Rubino},
  \citenamefont {Rozema}, \citenamefont {Feix}, \citenamefont {Araújo},
  \citenamefont {Zeuner}, \citenamefont {Procopio}, \citenamefont {Časlav
  Brukner},\ and\ \citenamefont {Walther}}]{Rubino2017}%
  \BibitemOpen
  \bibfield  {author} {\bibinfo {author} {\bibfnamefont {G.}~\bibnamefont
  {Rubino}}, \bibinfo {author} {\bibfnamefont {L.~A.}\ \bibnamefont {Rozema}},
  \bibinfo {author} {\bibfnamefont {A.}~\bibnamefont {Feix}}, \bibinfo {author}
  {\bibfnamefont {M.}~\bibnamefont {Araújo}}, \bibinfo {author} {\bibfnamefont
  {J.~M.}\ \bibnamefont {Zeuner}}, \bibinfo {author} {\bibfnamefont {L.~M.}\
  \bibnamefont {Procopio}}, \bibinfo {author} {\bibnamefont {Časlav
  Brukner}},\ and\ \bibinfo {author} {\bibfnamefont {P.}~\bibnamefont
  {Walther}},\ }\bibfield  {title} {\bibinfo {title} {Experimental verification
  of an indefinite causal order},\ }\href
  {https://doi.org/10.1126/sciadv.1602589} {\bibfield  {journal} {\bibinfo
  {journal} {Science Advances}\ }\textbf {\bibinfo {volume} {3}},\ \bibinfo
  {pages} {e1602589} (\bibinfo {year} {2017})}\BibitemShut {NoStop}%
\bibitem [{\citenamefont {Goswami}\ \emph {et~al.}(2018)\citenamefont
  {Goswami}, \citenamefont {Giarmatzi}, \citenamefont {Kewming}, \citenamefont
  {Costa}, \citenamefont {Branciard}, \citenamefont {Romero},\ and\
  \citenamefont {White}}]{Goswami2018}%
  \BibitemOpen
  \bibfield  {author} {\bibinfo {author} {\bibfnamefont {K.}~\bibnamefont
  {Goswami}}, \bibinfo {author} {\bibfnamefont {C.}~\bibnamefont {Giarmatzi}},
  \bibinfo {author} {\bibfnamefont {M.}~\bibnamefont {Kewming}}, \bibinfo
  {author} {\bibfnamefont {F.}~\bibnamefont {Costa}}, \bibinfo {author}
  {\bibfnamefont {C.}~\bibnamefont {Branciard}}, \bibinfo {author}
  {\bibfnamefont {J.}~\bibnamefont {Romero}},\ and\ \bibinfo {author}
  {\bibfnamefont {A.~G.}\ \bibnamefont {White}},\ }\bibfield  {title} {\bibinfo
  {title} {Indefinite causal order in a quantum switch},\ }\href
  {https://doi.org/10.1103/PhysRevLett.121.090503} {\bibfield  {journal}
  {\bibinfo  {journal} {Phys. Rev. Lett.}\ }\textbf {\bibinfo {volume} {121}},\
  \bibinfo {pages} {090503} (\bibinfo {year} {2018})}\BibitemShut {NoStop}%
\bibitem [{\citenamefont {Goswami}\ \emph {et~al.}(2020)\citenamefont
  {Goswami}, \citenamefont {Cao}, \citenamefont {Paz-Silva}, \citenamefont
  {Romero},\ and\ \citenamefont {White}}]{Goswami2020}%
  \BibitemOpen
  \bibfield  {author} {\bibinfo {author} {\bibfnamefont {K.}~\bibnamefont
  {Goswami}}, \bibinfo {author} {\bibfnamefont {Y.}~\bibnamefont {Cao}},
  \bibinfo {author} {\bibfnamefont {G.~A.}\ \bibnamefont {Paz-Silva}}, \bibinfo
  {author} {\bibfnamefont {J.}~\bibnamefont {Romero}},\ and\ \bibinfo {author}
  {\bibfnamefont {A.~G.}\ \bibnamefont {White}},\ }\bibfield  {title} {\bibinfo
  {title} {Increasing communication capacity via superposition of order},\
  }\href {https://doi.org/10.1103/PhysRevResearch.2.033292} {\bibfield
  {journal} {\bibinfo  {journal} {Phys. Rev. Res.}\ }\textbf {\bibinfo {volume}
  {2}},\ \bibinfo {pages} {033292} (\bibinfo {year} {2020})}\BibitemShut
  {NoStop}%
\bibitem [{\citenamefont {Deng}\ \emph {et~al.}(2025)\citenamefont {Deng},
  \citenamefont {Liu}, \citenamefont {Chen}, \citenamefont {Fu}, \citenamefont
  {Bao}, \citenamefont {Zheng}, \citenamefont {Gong}, \citenamefont {He},\ and\
  \citenamefont {Wang}}]{Deng2025}%
  \BibitemOpen
  \bibfield  {author} {\bibinfo {author} {\bibfnamefont {Y.}~\bibnamefont
  {Deng}}, \bibinfo {author} {\bibfnamefont {S.}~\bibnamefont {Liu}}, \bibinfo
  {author} {\bibfnamefont {X.}~\bibnamefont {Chen}}, \bibinfo {author}
  {\bibfnamefont {Z.}~\bibnamefont {Fu}}, \bibinfo {author} {\bibfnamefont
  {J.}~\bibnamefont {Bao}}, \bibinfo {author} {\bibfnamefont {Y.}~\bibnamefont
  {Zheng}}, \bibinfo {author} {\bibfnamefont {Q.}~\bibnamefont {Gong}},
  \bibinfo {author} {\bibfnamefont {Q.}~\bibnamefont {He}},\ and\ \bibinfo
  {author} {\bibfnamefont {J.}~\bibnamefont {Wang}},\ }\bibfield  {title}
  {\bibinfo {title} {Generalized indefinite causal orders in an integrated
  quantum switch},\ }\href {https://doi.org/10.1103/39vh-84n1} {\bibfield
  {journal} {\bibinfo  {journal} {Phys. Rev. Lett.}\ }\textbf {\bibinfo
  {volume} {135}},\ \bibinfo {pages} {160202} (\bibinfo {year}
  {2025})}\BibitemShut {NoStop}%
\bibitem [{\citenamefont {Chiribella}(2012)}]{Chiribella2012}%
  \BibitemOpen
  \bibfield  {author} {\bibinfo {author} {\bibfnamefont {G.}~\bibnamefont
  {Chiribella}},\ }\bibfield  {title} {\bibinfo {title} {Perfect discrimination
  of no-signalling channels via quantum superposition of causal structures},\
  }\href {https://doi.org/10.1103/PhysRevA.86.040301} {\bibfield  {journal}
  {\bibinfo  {journal} {Phys. Rev. A}\ }\textbf {\bibinfo {volume} {86}},\
  \bibinfo {pages} {040301} (\bibinfo {year} {2012})}\BibitemShut {NoStop}%
\bibitem [{\citenamefont {Ebler}\ \emph {et~al.}(2018)\citenamefont {Ebler},
  \citenamefont {Salek},\ and\ \citenamefont {Chiribella}}]{Ebler2018}%
  \BibitemOpen
  \bibfield  {author} {\bibinfo {author} {\bibfnamefont {D.}~\bibnamefont
  {Ebler}}, \bibinfo {author} {\bibfnamefont {S.}~\bibnamefont {Salek}},\ and\
  \bibinfo {author} {\bibfnamefont {G.}~\bibnamefont {Chiribella}},\ }\bibfield
   {title} {\bibinfo {title} {Enhanced communication with the assistance of
  indefinite causal order},\ }\href
  {https://doi.org/10.1103/PhysRevLett.120.120502} {\bibfield  {journal}
  {\bibinfo  {journal} {Phys. Rev. Lett.}\ }\textbf {\bibinfo {volume} {120}},\
  \bibinfo {pages} {120502} (\bibinfo {year} {2018})}\BibitemShut {NoStop}%
\bibitem [{\citenamefont {Chiribella}\ \emph {et~al.}(2021)\citenamefont
  {Chiribella}, \citenamefont {Banik}, \citenamefont {Bhattacharya},
  \citenamefont {Guha}, \citenamefont {Alimuddin}, \citenamefont {Roy},
  \citenamefont {Saha}, \citenamefont {Agrawal},\ and\ \citenamefont
  {Kar}}]{Chiribella2021}%
  \BibitemOpen
  \bibfield  {author} {\bibinfo {author} {\bibfnamefont {G.}~\bibnamefont
  {Chiribella}}, \bibinfo {author} {\bibfnamefont {M.}~\bibnamefont {Banik}},
  \bibinfo {author} {\bibfnamefont {S.~S.}\ \bibnamefont {Bhattacharya}},
  \bibinfo {author} {\bibfnamefont {T.}~\bibnamefont {Guha}}, \bibinfo {author}
  {\bibfnamefont {M.}~\bibnamefont {Alimuddin}}, \bibinfo {author}
  {\bibfnamefont {A.}~\bibnamefont {Roy}}, \bibinfo {author} {\bibfnamefont
  {S.}~\bibnamefont {Saha}}, \bibinfo {author} {\bibfnamefont {S.}~\bibnamefont
  {Agrawal}},\ and\ \bibinfo {author} {\bibfnamefont {G.}~\bibnamefont {Kar}},\
  }\bibfield  {title} {\bibinfo {title} {Indefinite causal order enables
  perfect quantum communication with zero capacity channels},\ }\href
  {https://doi.org/10.1088/1367-2630/abe7a0} {\bibfield  {journal} {\bibinfo
  {journal} {New Journal of Physics}\ }\textbf {\bibinfo {volume} {23}},\
  \bibinfo {pages} {033039} (\bibinfo {year} {2021})}\BibitemShut {NoStop}%
\bibitem [{\citenamefont {Coles}\ \emph {et~al.}(2014)\citenamefont {Coles},
  \citenamefont {Kaniewski},\ and\ \citenamefont {Wehner}}]{Coles2014}%
  \BibitemOpen
  \bibfield  {author} {\bibinfo {author} {\bibfnamefont {P.~J.}\ \bibnamefont
  {Coles}}, \bibinfo {author} {\bibfnamefont {J.}~\bibnamefont {Kaniewski}},\
  and\ \bibinfo {author} {\bibfnamefont {S.}~\bibnamefont {Wehner}},\
  }\bibfield  {title} {\bibinfo {title} {Equivalence of wave--particle duality
  to entropic uncertainty},\ }\href {https://doi.org/10.1038/ncomms6814}
  {\bibfield  {journal} {\bibinfo  {journal} {Nature Communications}\ }\textbf
  {\bibinfo {volume} {5}},\ \bibinfo {pages} {5814} (\bibinfo {year}
  {2014})}\BibitemShut {NoStop}%
\bibitem [{\citenamefont {Coles}\ \emph {et~al.}(2017)\citenamefont {Coles},
  \citenamefont {Berta}, \citenamefont {Tomamichel},\ and\ \citenamefont
  {Wehner}}]{Coles2017}%
  \BibitemOpen
  \bibfield  {author} {\bibinfo {author} {\bibfnamefont {P.~J.}\ \bibnamefont
  {Coles}}, \bibinfo {author} {\bibfnamefont {M.}~\bibnamefont {Berta}},
  \bibinfo {author} {\bibfnamefont {M.}~\bibnamefont {Tomamichel}},\ and\
  \bibinfo {author} {\bibfnamefont {S.}~\bibnamefont {Wehner}},\ }\bibfield
  {title} {\bibinfo {title} {Entropic uncertainty relations and their
  applications},\ }\href {https://doi.org/10.1103/RevModPhys.89.015002}
  {\bibfield  {journal} {\bibinfo  {journal} {Rev. Mod. Phys.}\ }\textbf
  {\bibinfo {volume} {89}},\ \bibinfo {pages} {015002} (\bibinfo {year}
  {2017})}\BibitemShut {NoStop}%
\bibitem [{\citenamefont {Paul}\ and\ \citenamefont
  {Qureshi}(2017)}]{PaulQureshi2017}%
  \BibitemOpen
  \bibfield  {author} {\bibinfo {author} {\bibfnamefont {T.}~\bibnamefont
  {Paul}}\ and\ \bibinfo {author} {\bibfnamefont {T.}~\bibnamefont {Qureshi}},\
  }\bibfield  {title} {\bibinfo {title} {Measuring quantum coherence in
  multislit interference},\ }\href {https://doi.org/10.1103/PhysRevA.95.042110}
  {\bibfield  {journal} {\bibinfo  {journal} {Phys. Rev. A}\ }\textbf {\bibinfo
  {volume} {95}},\ \bibinfo {pages} {042110} (\bibinfo {year}
  {2017})}\BibitemShut {NoStop}%
\bibitem [{\citenamefont {Afshar}\ \emph {et~al.}(2007)\citenamefont {Afshar},
  \citenamefont {Flores}, \citenamefont {McDonald},\ and\ \citenamefont
  {Knoesel}}]{afshar2007}%
  \BibitemOpen
  \bibfield  {author} {\bibinfo {author} {\bibfnamefont {S.~S.}\ \bibnamefont
  {Afshar}}, \bibinfo {author} {\bibfnamefont {E.}~\bibnamefont {Flores}},
  \bibinfo {author} {\bibfnamefont {K.~F.}\ \bibnamefont {McDonald}},\ and\
  \bibinfo {author} {\bibfnamefont {E.}~\bibnamefont {Knoesel}},\ }\bibfield
  {title} {\bibinfo {title} {Paradox in wave-particle duality},\ }\href
  {https://doi.org/10.1007/s10701-006-9102-8} {\bibfield  {journal} {\bibinfo
  {journal} {Foundations of Physics}\ }\textbf {\bibinfo {volume} {37}},\
  \bibinfo {pages} {295–305} (\bibinfo {year} {2007})}\BibitemShut {NoStop}%
\bibitem [{\citenamefont {Qureshi}(2012)}]{Qureshi2012}%
  \BibitemOpen
  \bibfield  {author} {\bibinfo {author} {\bibfnamefont {T.}~\bibnamefont
  {Qureshi}},\ }\bibfield  {title} {\bibinfo {title} {Modified two-slit
  experiments and complementarity},\ }\href
  {https://doi.org/10.4236/jqis.2012.22007} {\bibfield  {journal} {\bibinfo
  {journal} {Journal of Quantum Information Science}\ }\textbf {\bibinfo
  {volume} {02}},\ \bibinfo {pages} {35–40} (\bibinfo {year}
  {2012})}\BibitemShut {NoStop}%
\bibitem [{\citenamefont {Qureshi}(2023)}]{Qureshi2023}%
  \BibitemOpen
  \bibfield  {author} {\bibinfo {author} {\bibfnamefont {T.}~\bibnamefont
  {Qureshi}},\ }\bibfield  {title} {\bibinfo {title} {Understanding modified
  two-slit experiments using path markers},\ }\href
  {https://doi.org/10.1007/s10701-023-00684-z} {\bibfield  {journal} {\bibinfo
  {journal} {Foundations of Physics}\ }\textbf {\bibinfo {volume} {53}},\
  \bibinfo {pages} {38} (\bibinfo {year} {2023})}\BibitemShut {NoStop}%
\bibitem [{\citenamefont {Araújo}\ \emph {et~al.}(2015)\citenamefont
  {Araújo}, \citenamefont {Branciard}, \citenamefont {Costa}, \citenamefont
  {Feix}, \citenamefont {Giarmatzi},\ and\ \citenamefont
  {Brukner}}]{Araujo2015}%
  \BibitemOpen
  \bibfield  {author} {\bibinfo {author} {\bibfnamefont {M.}~\bibnamefont
  {Araújo}}, \bibinfo {author} {\bibfnamefont {C.}~\bibnamefont {Branciard}},
  \bibinfo {author} {\bibfnamefont {F.}~\bibnamefont {Costa}}, \bibinfo
  {author} {\bibfnamefont {A.}~\bibnamefont {Feix}}, \bibinfo {author}
  {\bibfnamefont {C.}~\bibnamefont {Giarmatzi}},\ and\ \bibinfo {author}
  {\bibfnamefont {{\v{C}}.}~\bibnamefont {Brukner}},\ }\bibfield  {title}
  {\bibinfo {title} {Witnessing causal nonseparability},\ }\href
  {https://doi.org/10.1088/1367-2630/17/10/102001} {\bibfield  {journal}
  {\bibinfo  {journal} {New Journal of Physics}\ }\textbf {\bibinfo {volume}
  {17}},\ \bibinfo {pages} {102001} (\bibinfo {year} {2015})}\BibitemShut
  {NoStop}%
\bibitem [{\citenamefont {Branciard}(2016)}]{Branciard2016}%
  \BibitemOpen
  \bibfield  {author} {\bibinfo {author} {\bibfnamefont {C.}~\bibnamefont
  {Branciard}},\ }\bibfield  {title} {\bibinfo {title} {Witnesses of causal
  nonseparability: an introduction and a few case studies},\ }\href
  {https://doi.org/10.1038/srep26018} {\bibfield  {journal} {\bibinfo
  {journal} {Scientific Reports}\ }\textbf {\bibinfo {volume} {6}},\ \bibinfo
  {pages} {26018} (\bibinfo {year} {2016})}\BibitemShut {NoStop}%
\bibitem [{\citenamefont {Costa}(2022)}]{Costa2022}%
  \BibitemOpen
  \bibfield  {author} {\bibinfo {author} {\bibfnamefont {F.}~\bibnamefont
  {Costa}},\ }\bibfield  {title} {\bibinfo {title} {A no-go theorem for
  superpositions of causal orders},\ }\href
  {https://doi.org/10.22331/q-2022-03-01-663} {\bibfield  {journal} {\bibinfo
  {journal} {{Quantum}}\ }\textbf {\bibinfo {volume} {6}},\ \bibinfo {pages}
  {663} (\bibinfo {year} {2022})}\BibitemShut {NoStop}%
\bibitem [{\citenamefont {Scully}\ and\ \citenamefont
  {Dr\"uhl}(1982)}]{Scully1982}%
  \BibitemOpen
  \bibfield  {author} {\bibinfo {author} {\bibfnamefont {M.~O.}\ \bibnamefont
  {Scully}}\ and\ \bibinfo {author} {\bibfnamefont {K.}~\bibnamefont
  {Dr\"uhl}},\ }\bibfield  {title} {\bibinfo {title} {Quantum eraser: A
  proposed photon correlation experiment concerning observation and ``delayed
  choice" in quantum mechanics},\ }\href
  {https://doi.org/10.1103/PhysRevA.25.2208} {\bibfield  {journal} {\bibinfo
  {journal} {Phys. Rev. A}\ }\textbf {\bibinfo {volume} {25}},\ \bibinfo
  {pages} {2208} (\bibinfo {year} {1982})}\BibitemShut {NoStop}%
\bibitem [{\citenamefont {Walborn}\ \emph {et~al.}(2002)\citenamefont
  {Walborn}, \citenamefont {Terra~Cunha}, \citenamefont {P\'adua},\ and\
  \citenamefont {Monken}}]{Walborn2002}%
  \BibitemOpen
  \bibfield  {author} {\bibinfo {author} {\bibfnamefont {S.~P.}\ \bibnamefont
  {Walborn}}, \bibinfo {author} {\bibfnamefont {M.~O.}\ \bibnamefont
  {Terra~Cunha}}, \bibinfo {author} {\bibfnamefont {S.}~\bibnamefont
  {P\'adua}},\ and\ \bibinfo {author} {\bibfnamefont {C.~H.}\ \bibnamefont
  {Monken}},\ }\bibfield  {title} {\bibinfo {title} {Double-slit quantum
  eraser},\ }\href {https://doi.org/10.1103/PhysRevA.65.033818} {\bibfield
  {journal} {\bibinfo  {journal} {Phys. Rev. A}\ }\textbf {\bibinfo {volume}
  {65}},\ \bibinfo {pages} {033818} (\bibinfo {year} {2002})}\BibitemShut
  {NoStop}%
\bibitem [{\citenamefont {Siddiqui}\ and\ \citenamefont
  {Qureshi}(2016)}]{SiddiquiQureshi2016}%
  \BibitemOpen
  \bibfield  {author} {\bibinfo {author} {\bibfnamefont {M.~A.}\ \bibnamefont
  {Siddiqui}}\ and\ \bibinfo {author} {\bibfnamefont {T.}~\bibnamefont
  {Qureshi}},\ }\bibfield  {title} {\bibinfo {title} {A nonlocal wave--particle
  duality},\ }\href {https://doi.org/10.1007/s40509-015-0064-4} {\bibfield
  {journal} {\bibinfo  {journal} {Quantum Studies: Mathematics and
  Foundations}\ }\textbf {\bibinfo {volume} {3}},\ \bibinfo {pages} {115}
  (\bibinfo {year} {2016})}\BibitemShut {NoStop}%
\bibitem [{\citenamefont {Ivanovic}(1987)}]{Ivanovic1987}%
  \BibitemOpen
  \bibfield  {author} {\bibinfo {author} {\bibfnamefont {I.}~\bibnamefont
  {Ivanovic}},\ }\bibfield  {title} {\bibinfo {title} {How to differentiate
  between non-orthogonal states},\ }\href
  {https://doi.org/https://doi.org/10.1016/0375-9601(87)90222-2} {\bibfield
  {journal} {\bibinfo  {journal} {Physics Letters A}\ }\textbf {\bibinfo
  {volume} {123}},\ \bibinfo {pages} {257} (\bibinfo {year}
  {1987})}\BibitemShut {NoStop}%
\bibitem [{\citenamefont {Dieks}(1988)}]{Dieks1988}%
  \BibitemOpen
  \bibfield  {author} {\bibinfo {author} {\bibfnamefont {D.}~\bibnamefont
  {Dieks}},\ }\bibfield  {title} {\bibinfo {title} {Overlap and
  distinguishability of quantum states},\ }\href
  {https://doi.org/https://doi.org/10.1016/0375-9601(88)90840-7} {\bibfield
  {journal} {\bibinfo  {journal} {Physics Letters A}\ }\textbf {\bibinfo
  {volume} {126}},\ \bibinfo {pages} {303} (\bibinfo {year}
  {1988})}\BibitemShut {NoStop}%
\bibitem [{\citenamefont {Peres}(1988)}]{Peres1988}%
  \BibitemOpen
  \bibfield  {author} {\bibinfo {author} {\bibfnamefont {A.}~\bibnamefont
  {Peres}},\ }\bibfield  {title} {\bibinfo {title} {How to differentiate
  between non-orthogonal states},\ }\href
  {https://doi.org/https://doi.org/10.1016/0375-9601(88)91034-1} {\bibfield
  {journal} {\bibinfo  {journal} {Physics Letters A}\ }\textbf {\bibinfo
  {volume} {128}},\ \bibinfo {pages} {19} (\bibinfo {year} {1988})}\BibitemShut
  {NoStop}%
\bibitem [{\citenamefont {Jaeger}\ and\ \citenamefont
  {Shimony}(1995)}]{Jaeger1995}%
  \BibitemOpen
  \bibfield  {author} {\bibinfo {author} {\bibfnamefont {G.}~\bibnamefont
  {Jaeger}}\ and\ \bibinfo {author} {\bibfnamefont {A.}~\bibnamefont
  {Shimony}},\ }\bibfield  {title} {\bibinfo {title} {Optimal distinction
  between two non-orthogonal quantum states},\ }\href
  {https://doi.org/https://doi.org/10.1016/0375-9601(94)00919-G} {\bibfield
  {journal} {\bibinfo  {journal} {Physics Letters A}\ }\textbf {\bibinfo
  {volume} {197}},\ \bibinfo {pages} {83} (\bibinfo {year} {1995})}\BibitemShut
  {NoStop}%
\bibitem [{\citenamefont {Jakob}\ and\ \citenamefont
  {Bergou}(2010)}]{Jakob2010}%
  \BibitemOpen
  \bibfield  {author} {\bibinfo {author} {\bibfnamefont {M.}~\bibnamefont
  {Jakob}}\ and\ \bibinfo {author} {\bibfnamefont {J.~A.}\ \bibnamefont
  {Bergou}},\ }\bibfield  {title} {\bibinfo {title} {Quantitative
  complementarity relations in bipartite systems: Entanglement as a physical
  reality},\ }\href
  {https://doi.org/https://doi.org/10.1016/j.optcom.2009.10.044} {\bibfield
  {journal} {\bibinfo  {journal} {Optics Communications}\ }\textbf {\bibinfo
  {volume} {283}},\ \bibinfo {pages} {827} (\bibinfo {year}
  {2010})}\BibitemShut {NoStop}%
\bibitem [{\citenamefont {Qian}\ \emph {et~al.}(2018)\citenamefont {Qian},
  \citenamefont {Vamivakas},\ and\ \citenamefont {Eberly}}]{Qian2018}%
  \BibitemOpen
  \bibfield  {author} {\bibinfo {author} {\bibfnamefont {X.-F.}\ \bibnamefont
  {Qian}}, \bibinfo {author} {\bibfnamefont {A.~N.}\ \bibnamefont
  {Vamivakas}},\ and\ \bibinfo {author} {\bibfnamefont {J.~H.}\ \bibnamefont
  {Eberly}},\ }\bibfield  {title} {\bibinfo {title} {Entanglement limits
  duality and vice versa},\ }\href {https://doi.org/10.1364/OPTICA.5.000942}
  {\bibfield  {journal} {\bibinfo  {journal} {Optica}\ }\textbf {\bibinfo
  {volume} {5}},\ \bibinfo {pages} {942} (\bibinfo {year} {2018})}\BibitemShut
  {NoStop}%
\bibitem [{\citenamefont {Berta}\ \emph {et~al.}(2010)\citenamefont {Berta},
  \citenamefont {Christandl}, \citenamefont {Colbeck}, \citenamefont {Renes},\
  and\ \citenamefont {Renner}}]{Berta2010}%
  \BibitemOpen
  \bibfield  {author} {\bibinfo {author} {\bibfnamefont {M.}~\bibnamefont
  {Berta}}, \bibinfo {author} {\bibfnamefont {M.}~\bibnamefont {Christandl}},
  \bibinfo {author} {\bibfnamefont {R.}~\bibnamefont {Colbeck}}, \bibinfo
  {author} {\bibfnamefont {J.~M.}\ \bibnamefont {Renes}},\ and\ \bibinfo
  {author} {\bibfnamefont {R.}~\bibnamefont {Renner}},\ }\bibfield  {title}
  {\bibinfo {title} {The uncertainty principle in the presence of quantum
  memory},\ }\href {https://doi.org/10.1038/nphys1734} {\bibfield  {journal}
  {\bibinfo  {journal} {Nature Physics}\ }\textbf {\bibinfo {volume} {6}},\
  \bibinfo {pages} {659} (\bibinfo {year} {2010})}\BibitemShut {NoStop}%
\bibitem [{\citenamefont {Chefles}(2000)}]{Chefles2000}%
  \BibitemOpen
  \bibfield  {author} {\bibinfo {author} {\bibfnamefont {A.}~\bibnamefont
  {Chefles}},\ }\bibfield  {title} {\bibinfo {title} {Quantum state
  discrimination},\ }\href {https://doi.org/10.1080/00107510010002599}
  {\bibfield  {journal} {\bibinfo  {journal} {Contemporary Physics}\ }\textbf
  {\bibinfo {volume} {41}},\ \bibinfo {pages} {401} (\bibinfo {year}
  {2000})}\BibitemShut {NoStop}%
\bibitem [{\citenamefont {Bergou}(2007)}]{Bergou2007}%
  \BibitemOpen
  \bibfield  {author} {\bibinfo {author} {\bibfnamefont {J.~A.}\ \bibnamefont
  {Bergou}},\ }\bibfield  {title} {\bibinfo {title} {Quantum state
  discrimination and selected applications},\ }\href
  {https://doi.org/10.1088/1742-6596/84/1/012001} {\bibfield  {journal}
  {\bibinfo  {journal} {Journal of Physics: Conference Series}\ }\textbf
  {\bibinfo {volume} {84}},\ \bibinfo {pages} {012001} (\bibinfo {year}
  {2007})}\BibitemShut {NoStop}%
\bibitem [{\citenamefont {Boyd}\ and\ \citenamefont
  {Vandenberghe}(2004)}]{Boyd2004}%
  \BibitemOpen
  \bibfield  {author} {\bibinfo {author} {\bibfnamefont {S.}~\bibnamefont
  {Boyd}}\ and\ \bibinfo {author} {\bibfnamefont {L.}~\bibnamefont
  {Vandenberghe}},\ }\href@noop {} {\emph {\bibinfo {title} {Convex
  Optimization}}}\ (\bibinfo  {publisher} {Cambridge University Press},\
  \bibinfo {year} {2004})\BibitemShut {NoStop}%
\bibitem [{\citenamefont {Helstrom}(1976)}]{Helstrom1976}%
  \BibitemOpen
  \bibfield  {author} {\bibinfo {author} {\bibfnamefont {C.~W.}\ \bibnamefont
  {Helstrom}},\ }\href@noop {} {\emph {\bibinfo {title} {Quantum Detection and
  Estimation Theory}}}\ (\bibinfo  {publisher} {Academic Press},\ \bibinfo
  {address} {New York},\ \bibinfo {year} {1976})\BibitemShut {NoStop}%
\bibitem [{\citenamefont {Bergou}(2010)}]{Bergou2010}%
  \BibitemOpen
  \bibfield  {author} {\bibinfo {author} {\bibfnamefont {J.~A.}\ \bibnamefont
  {Bergou}},\ }\bibfield  {title} {\bibinfo {title} {Discrimination of quantum
  states},\ }\href {https://doi.org/10.1080/09500340903477756} {\bibfield
  {journal} {\bibinfo  {journal} {Journal of Modern Optics}\ }\textbf {\bibinfo
  {volume} {57}},\ \bibinfo {pages} {160} (\bibinfo {year} {2010})}\BibitemShut
  {NoStop}%
\end{thebibliography}
